\documentclass[11pt,letter]{article}
\usepackage{lineno}
\usepackage{multirow} 
\usepackage{multicol}  
\usepackage{amsmath,amssymb,amsthm}
\usepackage{thmtools}
\usepackage{thm-restate}
\usepackage[margin=1in]{geometry}
\usepackage[utf8]{inputenc}
\usepackage{booktabs}
\usepackage{amsmath}
\usepackage[colorlinks=true,linkcolor=red,citecolor=blue]{hyperref}
\usepackage{amssymb}
\usepackage{tikz}
\usetikzlibrary{topaths,calc}
\usepackage{thm-restate,thmtools}
\usepackage{mathtools}
\usepackage{amsthm}
\usepackage{changepage} 
\usepackage{booktabs} 
\usepackage{caption}  
\usepackage{graphicx} 
\usepackage{siunitx}  
\usepackage{xcolor} 
\usepackage{lipsum}     

\usepackage{framed} 

\newenvironment{introquestion}[1] 
{
	\begin{center} 
		\begin{minipage}{0.9\textwidth} 
			\begin{framed} 
				\textbf{#1} \hspace{1ex} 
				\itshape  
			}
			{
			\end{framed}
		\end{minipage}
	\end{center}
}

\newenvironment{introInformal}[1]
{
	\noindent
	\begin{framed}
	
		\begin{minipage}{\dimexpr\linewidth-2\fboxsep-2\fboxrule\relax}  
				\centering\textbf{#1} \\[1ex]
			\itshape
			\raggedright
			\begin{adjustwidth}{-1.3em}{0em}
			}
			{
			\end{adjustwidth}
		\end{minipage}
	\end{framed}
}

\usepackage{color}
\usepackage{amsmath}
\usepackage{amsfonts}
\usepackage{amssymb}
\usepackage{bbm}
\usepackage{relsize}
\usepackage{enumitem}
\usepackage[procnumbered, linesnumbered,algoruled,noend]{algorithm2e}
\usepackage{lineno}
\usepackage{array}
\usepackage[T1]{fontenc}
\usepackage[utf8]{inputenc}

\usepackage[nameinlink]{cleveref}

\newcommand{\ceil}[1]{{\left\lceil#1  \right\rceil}}
\newcommand{\comment}[1]{}

\newcommand{\bases}{\textnormal{bases}}

\newcommand{\cA}{{\mathcal{A}}}
\newcommand{\cB}{{\mathcal{B}}}
\newcommand{\cQ}{{\mathcal{Q}}}
\newcommand{\cG}{{\mathcal{G}}}
\newcommand{\cI}{{\mathcal{I}}}
\newcommand{\cS}{{\mathcal{S}}}

\newcommand{\cD}{{\mathcal{D}}}

\newcommand{\cF}{{\mathcal{F}}}
\newcommand{\cH}{{\mathcal{H}}}
\newcommand{\cX}{{\mathcal{X}}}

\newcommand{\sample}{\textnormal{\textsc{Sample}}}
\newcommand{\cL}{\mathcal{L}}
\newcommand{\cT}{\mathcal{T}}

\newcommand{\ext}{\textnormal{\textsf{Ext}}}
\newcommand{\oracle}{\textnormal{\textsf{oracle}}}
\newcommand{\Xspace}{{~}}
\newcommand{\cP}{{\mathcal{P}}}

\newcommand{\vb}{{\textnormal{\textsf{b}}}}

\newcommand{\ve}{{\textnormal{\textsf{e}}}}
\newcommand{\gr}{{\textnormal{\textsf{grid}}}}

\newcommand{\eps}{{\varepsilon}}

\newcommand{\N}{{\mathbb{N}}}
\newcommand{\floor}[1]{\left\lfloor #1 \right\rfloor}
\DeclareMathOperator*{\argmax}{arg\,max}
\DeclareMathOperator*{\argmin}{arg\,min}

\newcommand{\abs}[1]{{\left| #1\right|}}
\newcommand{\poly}{\textnormal{poly}}
\newcommand{\entropy}[1]{\mathcal{H}\left(#1 \right)}

\newcommand{\problem}[2]{
	\noindent
	\begin{tabular}{ p{0.8in} p{5.4in} }
		\\
		\hline
		\multicolumn{2}{c}{#1} \\
		\hline
		#2\\
		\hline
		\\
	\end{tabular}
}

\newtheorem{lemma}{Lemma}[section]

\newtheorem{definition}[lemma]{Definition}
\newtheorem{theorem}[lemma]{Theorem}

\newtheorem{claim}[lemma]{Claim}
\crefname{claim}{claim}{claims}

\newcommand{\one}{\mathbbm{1}}
	\crefname{claim}{claim}{claims}
		\crefname{cor}{corollary}{corollaries}
\begin{document}

		\title{
			 You (Almost) Can't Beat Brute Force for $3$-Matroid Intersection
		}
		\date{}
		\author{Ilan Doron-Arad\thanks{Computer Science Department, 
			Technion, Haifa, Israel. \texttt{ilan.d.a.s.d@gmail.com}}
		\and
		Ariel Kulik\thanks{Department of Industrial Engineering and Management at Ben-Gurion University. \texttt{ariel.kulik@gmail.com}} 
		\and
		Hadas Shachnai\thanks{Computer Science Department, 
			Technion, Haifa, Israel. \texttt{hadas@cs.technion.ac.il}}
	}
		\maketitle
		\thispagestyle{empty}
		
		\begin{abstract}
			
The {\em $\ell$-matroid intersection} ($\ell$-MI) problem asks if $\ell$ given matroids share a common basis.  Already for $\ell = 3$, notable canonical NP-complete special cases are $3$-Dimensional Matching and Hamiltonian Path on directed graphs. 
 However, while these problems admit exponential-time algorithms that improve the simple brute force significantly (e.g., Eiben-Koana-Wahlstr{\"o}m (SODA'24)), the fastest known algorithm for $3$-MI on general matroids is exactly brute force with runtime $2^{n}/\poly(n)$, where $n$ is the number of elements.  
Our main result shows that, in fact, brute force cannot be significantly improved, by ruling out an algorithm for $\ell$-MI with runtime $o\left(2^{n-5 \cdot n^{\frac{1}{\ell-1}} \cdot \log (n)}\right)$, for any fixed $\ell\geq 3$.
For $3$-MI, this gives a lower bound of $o \left(2^{n-5 \cdot \sqrt{n} \cdot \log (n)} \right)$. 
 
 Our negative result raises the following natural questions: (i) Is there an algorithm for $3$-MI with runtime strictly better than brute force?
 (ii) Can we separate the parameterized complexity of  $3$-MI from the important special case on linear matroids (parameterized by the rank of the matroids $k$)? In particular, can a lower bound match the existing $c^{k^2} \cdot \poly(n)$ algorithm of Huang-Ward (SIDMA'23) for general $\ell$-MI parameterized by the rank?
  
 We make progress towards obtaining affirmative answers to the above questions. In particular, we present (i) an algorithm which solves $\ell$-MI faster than brute force in time $2^{n-\Omega\left(\log^2 (n)\right)} $ for any $\ell \geq 3$, and (ii) a parameterized running time lower bound of $2^{(\ell-2) \cdot k \cdot \log k} \cdot \poly(n)$ for $\ell$-MI, for any $\ell \geq 3$.
  We obtain these results by generalizing the Monotone Local Search technique of Fomin-Gaspers-Lokshtanov-Saurabh (J. ACM'19).
Broadly speaking, given a {\em subset problem}, our generalization transforms any algorithm parameterized by solution size, with runtime of the form $f(k) \cdot \poly(n)$,
 into an exponential-time algorithm with runtime depending on~$f$.
  This implies that {\em any} $f(k) \cdot \poly(n)$ time parameterized algorithm for a subset problem yields a $2^{n-\omega(\log n)}$ time algorithm beating brute force, which may be of independent interest.

		\end{abstract}

\newpage

\thispagestyle{empty} 

		\tableofcontents

\newpage 

\pagenumbering{arabic} 

\section{Introduction}

Matroids are a central abstraction of fundamental combinatorial structures such as spanning trees and linear independence in vector spaces. Despite their generic attributes, they exhibit desired tractability for fundamental algorithmic problems, contributing to the extensive research on matroids since the beginning of the 20th century \cite{whitney1935t,tutte1959matroids,white1986theory,oxley2006matroid,welsh2010matroid}. A {\em matroid} can be defined as a set system $(E, \cI)$, where $E$ is a finite set and $\cI \subseteq 2^E$ are the {\em independent sets}, which satisfy the following. $(i)$ $\emptyset \in \cI$, $(ii)$ ({\em hereditary property}): for all $A \in \cI$ and $B \subseteq A$ it holds that $B \in \cI$, and $(iii)$ ({\em exchange property}): for all $A,B \in \cI$ where $|A| > |B|$, there is $e \in A \setminus B$ such that $B \cup \{e\} \in \cI$.\footnote{Several other characterizations exist (see, e.g., \cite{schrijver2003combinatorial}), indicating how natural is the notion of matroids.}

Matroids are especially interesting in the study of algorithms. Perhaps the most fundamental problem on matroids is finding an optimal (maximum or minimum) weight {\em basis}, where a basis is an inclusion-wise maximal independent set. 
The simple {\em greedy} algorithm, which iteratively chooses the best local improvement, 
 finds an optimal basis of a matroid. In fact, matroids are precisely the structures on which the greedy algorithm is optimal \cite{rado1957note,gale1968optimal,edmonds1971matroids}.  The matroid {\em intersection} problem, of finding a common basis of two matroids, is more difficult yet polynomially solvable \cite{lawler1975matroid,edmonds1979matroid,edmonds2003submodular,edmonds2010matroid}. On the other hand, $3$-{\em matroid intersection} ($3$-MI), the problem of deciding whether three matroids share a common basis, is known to be computationally more challenging. 
 The $\ell$-matroid intersection problem, for $\ell  \geq 3$, is the main problem studied in this paper. For $k \in \N$, let $[k] = \{1, 2, \ldots ,k\}$.
 
 \problem{	\label{def:3MI}{ 
 $\ell$-matroid intersection ($\ell$-MI)}}{
 	{\bf Instance} &  
 	$(E,\cI_1,\ldots,\cI_{\ell})$ such that $(E,\cI_i)$ is a matroid for all $i \in [\ell]$.\\
 	{\bf Objective} & Decide if there is a common basis of the $\ell$ matroids. 
 }
 We assume that the given matroids are accessed via membership oracles (see \Cref{sec:prel}); therefore, we refer to the problem as {\em oracle} $\ell$-MI. 
 
We note that $\ell$-MI captures as special cases canonical NP-complete problems already when $\ell=3$. This includes Hamiltonian Path (on directed graphs) and $3$-{\em dimensional matching} ($3$-DM) (see \Cref{sec:matroidOverview}). Hence,
$3$-MI does not admit a polynomial time algorithm unless P=NP \cite{karp2010reducibility}.
To quote E.L. Lawler, solving $3$-MI ``would reduce the whole panoply of combinatorial
optimization at our feet'' \cite{camerini1975bounds}.
 
 Fortunately, there is a silver lining: on {\em linear} matroids 
   (including $3$-DM and Hamiltonian Path),
 $\ell$-MI can be solved by a faster-than-brute-force algorithm. Specifically, Eiben et al.~\cite{eiben2024determinantal} recently showed that $\ell$-MI, on linear matroids represented by a common field, can be solved in time $2^{O(\ell\cdot k)}$, where $k$ is the cardinality of a basis.\footnote{All bases of a matroid have the same cardinality, called the {\em rank} of the matroid \cite{schrijver2003combinatorial}.} 
 This, in conjunction with the results of Fomin et al. \cite{FGLS19}, leads to an algorithm whose runtime is $O\left(c^n\right)$ for $c < 2$, where $n$ is the size of the ground set, for any fixed $\ell\in \mathbb{N}$.
  However, for general $3$-MI instances, aside from some heuristics (e.g., \cite{camerini1975bounds,camerini1978heuristically}), no known algorithm has a running time better than the simple brute-force algorithm, which enumerates over all $k$-subsets of the ground set, where $k$ is the rank. 
  Alas, this algorithm runs in time 
 $\Omega \left(\frac{2^n}{n}\right)$, where $n$ is the size of the ground set. Whether this algorithm can be qualitatively improved is the first question addressed in this paper.
 
 \begin{introquestion}{Question 1:}
 		Can $3$-matroid intersection be solved in time $O \left(c^{n} \right)$ for $c < 2$?
 \end{introquestion}

\subsection{Our Results}
\renewcommand{\arraystretch}{1.5} 

We start by answering Question 1  and then present an algorithm for $\ell$-MI which beats the brute force algorithm by a super-polynomial factor.
Finally, we obtain lower bounds for $3$-MI {\em parameterized} by solution size.
Notably, our generalization of the {\em monotone local search} technique of Fomin et al. \cite{FGLS19} plays a key role in obtaining our results. In particular, the 
resulting exponential-time algorithm crucially justifies our better-than-brute-force lower bound. We summarize our main results in \Cref{tab:results}.

\begin{table}[h!]
	\centering
	\begin{tabular}{!{\boldmath\vrule width 1.5pt}>{\centering\arraybackslash}p{3cm}!{\boldmath\vrule width 1.5pt}>{\centering\arraybackslash}p{3cm}!{\boldmath\vrule width 1.5pt}>{\centering\arraybackslash}p{3.5cm}|}
		\noalign{\hrule height 1.5pt}
		\multirow{2}{*}{\textbf{Lower bounds}} 
		& Exp-time & $\mathbf{2^{\Big(n-5 \cdot \sqrt{n} \cdot \log (n)\Big)}}$    \\ \cline{2-3}
		& Parameterized & $\mathbf{c^{k \cdot \log k} \cdot \poly(n)}$    \\ \noalign{\hrule height 1.5pt}
		\textbf{Algorithms} 
		& Exp-time & $\mathbf{2^{n-\Omega(log^2 n)} \cdot \poly(n)}$ \\ \hline
	\end{tabular}
	\caption{\label{tab:results} A summary of our main results. In the table, $n$ is the size of the ground set, $k$ is the cardinality of a solution, and $c$ is some constant.}
\end{table}

 Our first result answers Question\Xspace1 negatively, even if randomization is allowed. 

\begin{restatable}{theorem}{thmMI}
	\label{thm:3MIMain}
	For any $\ell \geq 3$ and $n\in \mathbb{N}$ such that $n\geq 2^{\ell-1}$ and $n^{\frac{1}{\ell-1}}\in \mathbb{N}$,  there is no randomized algorithm which decides oracle $\ell$-matroid intersection in fewer than $2^{\Big(n-5 \cdot n^{\frac{1}{\ell-1}} \cdot \log (n)\Big)}$ queries to the membership oracles, 
		on instances with $n$ elements. 
\end{restatable}

Specifically, for infinitely many integers $n$, the theorem gives an explicit lower bound on the number of queries an algorithm for $\ell$-matroid intersection must use.  
As the minimum number of queries is a lower bound on the overall running time, 
the above theorem implies there is no randomized algorithm which decides oracle $3$-matroid intersection in time $o \left(2^{n-5 \cdot \sqrt{n}\cdot \log n}\right)$. In particular, for every $c < 2$ there is no algorithm that decides $3$-MI in time $O \left(c^n\right)$ (thus answering Question 1). This lower bound is unconditional and substantially stronger than the known bounds of  $O \left(c^n\right)$ for a fixed $c < 2$. Such lower bounds follow from the special case of $3$-MI of common bases partitioning \cite{HorschIMOS24,berczi2021complexity},\footnote{The results of \cite{HorschIMOS24} interestingly show that the closely related {\em exact matroid intersection} on general matroids is intractable, answering an open question of \cite{Eisenbrand2024}.} or $3$-MI on specific classes of matroids (e.g., partition matroids \cite{kushagra2019semi,kushagra2020three}). 

To obtain the lower bound in \Cref{thm:3MIMain}, we use one non-linear matroid along with $\ell - 1$ linear matroids. The idea of incorporating a non-linear matroid in this context is not new and is inspired by a family of \emph{paving} matroids that have been used to establish lower bounds for other matroid problems~\cite{jensen1982complexity,lovasz1980matroid,soto2011simple,berczi2021complexity,doron2024lower,HorschIMOS24}. Notably, some of these constructions~\cite{HorschIMOS24,berczi2021complexity} already yield exponential lower bounds for $3$-MI of the form $O(c^n)$ for some fixed $c < 2$. However, resolving Question 1 requires a stronger framework.
In this work, we develop a more general construction that combines new combinatorial insights, enabling us to obtain significantly stronger bounds for $\ell$-MI. We describe our techniques in greater detail in \Cref{sec:techniques}.

Theorem \Cref{thm:3MIMain} gives unconditional tight bounds for $\ell$-MI in the oracle model. Theoretically, these results may not apply to algorithms for $\ell$-MI in the standard computational model, in which the matroids are encoded as part of the input. To show that our hardness results are not restricted to the oracle model, in the full version of the paper we present analogous lower bounds in the standard computational model, based on known complexity assumptions. 

\subsubsection{Relation to Parameterized Complexity}

\Cref{thm:3MIMain} gives a nearly tight lower bound for $\ell$-MI for every $\ell\geq 3$. Yet,  for many NP-hard problems it is possible to devise {\em parameterized} algorithms which are efficient on instances with small parameter values, including $\ell$-MI on linear matroids, which can be solved in time $2^{k} \cdot \poly(|E|)$ \cite{eiben2024determinantal}, where $k$ is the rank. 
This raises the following natural questions.
\begin{itemize} 
	\item[(i)] 
	Is there an algorithm for $\ell$-MI that runs faster than a brute force algorithm?
	More specifically, 
	can parameterized algorithms lead to faster than brute-force algorithms  
	for $\ell$-MI? 

	\item[(ii)] Can we separate the parameterized complexity of $\ell$-MI from $\ell$-MI on linear matroids?

\end{itemize}  

We answer both questions affirmatively, using a generalization of the {\em Monotone Local Search} technique of Fomin et al. \cite{FGLS19}.

Monotone Local Search tackles a wide range of problems which can be cast as {\em implicit set systems}. In such problems the input encodes a set system $(E,\cF)$, where $E$ is an arbitrary ground set, and $\cF\subseteq 2^{E}$ is a collection of subsets of the ground set. 
It is assumed that the ground set $E$ can be computed from the input in polynomial time, and it can be determined if $S\in \cF$ for every $S\subseteq E$ in polynomial time. 
The objective is to decide if $\cF=\emptyset$. Numerous problems, including Vertex Cover, Feedback Vertex Set and Multicut on Trees, can be cast as implicit set systems  (see \cite{FGLS19} for additional problems).

For example, the input for {\em Vertex Cover} is a graph $G$ and a number $k\in\mathbb{N}$. The input encodes the set system $(E,\cF)$, where $E$ is the set of {\em vertices} of $G$  and $\cF$ is the set of all the {\em vertex covers} of $G$ of size $k$ ($S$ is a vertex cover if for every edge $(u,v)$ of $G$ it holds that $u\in S$ or $v\in S$).  Then, $G$ has a vertex cover of size $k$ if and only if $\cF\neq \emptyset$. Furthermore, the set $E$ can be easily computed given the graph $G$, and it is possible to determine if $S\in \cF$ in polynomial time.  That is,  Vertex Cover can be cast as an implicit set system. 

Fomin et al.~\cite{FGLS19} show how to convert an {\em extension algorithm} for an implicit set system $\cP$ into an exponential time  algorithm for $\cP$. 
An extension algorithm of time $c^k$ takes a $\cP$ instance $I$ as input, along with $X\subseteq E$ and a number $k$. 
The algorithm either returns a set $S\subseteq E\setminus X$ such that $X\cup S\in \cF$  and $|S|= k$, or decides that no such set exists,  where $(E,\cF)$ is the set system associated with the instance $I$.   That is, its goal is to extend $X$ into a set in $\cF$. 
Furthermore, the algorithm runs in time $c^k\cdot  \poly(|I|)$. It is often the case that extension algorithms can be easily derived from {\em parameterized} algorithms for the same problem. 

The main result of \cite{FGLS19} is that if an implicit set system has an extension algorithm of time $c^k\cdot \poly(|I|)$, then there is an algorithm for the same implicit set system which runs in time $\left(2-\frac{1}{c}\right)^n\cdot \poly(|I|)$, where $n=|E|$ is the size of the ground set associated with the instance. 
For example, this means that the parameterized $1.252^k\cdot\poly(|I|)$ algorithm for Vertex Cover \cite{HN24} enables to obtain a $\left(2-\frac{1}{1.252}\right)^n \cdot \poly(|I|) \approx 1.2012^n\cdot \poly(|I|)$ algorithm for the problem, where $n$ is the number of vertices in the graph.  This holds as the parameterized algorithm for Vertex Cover can be easily used to derive an extension algorithm for the problem. 
The technique has been extended to approximation algorithms \cite{EKMNS22,EKMNS23,EKMNS24} and to multivariate subroutines \cite{GL17}.  It was also shown in \cite{EKMNS24} that the conversion done by the technique is, in some sense, tight.

A {\em parameterized} randomized algorithm for oracle $\ell$-matroid intersection is a randomized algorithm for oracle-$\ell$-MI which runs in time $f(k)\cdot \poly(n)$, where $n=|E|$ is the size of the ground set of the input instance, and $k$ is the rank of $(E,\cI_1)$, the first matroid of the instance. Intuitively, such an algorithm is efficient if $k$ is significantly smaller than $n$. 
An algorithm for $\ell$-MI parameterized by $k$, whose runtime is  $c^{k^2}\cdot \poly(|E|)$, for some $c\geq 1$,  has been proposed by Huang and Ward \cite{huang2023fpt} for every $\ell\geq 3$, and the abovementioned result of Eiben et al. \cite{eiben2024determinantal} gives a faster running time of $2^{k} \cdot \poly(|E|)$ for linear matroids.  We refer the reader to standard textbooks \cite{cygan2015parameterized,downey2012parameterized} for a comprehensive introduction to parameterized complexity and more general definitions.

Up to some  technicalities due to the oracles, it is possible to cast $\ell$-MI as an implicit set system.
Furthermore, it is easy to show that a parameterized algorithm for $3$-MI of time $g(k)\cdot \poly(|E|)$  implies an extension algorithm  of the same running time. 
 Therefore, by \cite{FGLS19}, a parameterized $c^k\cdot \poly(|E|)$ algorithm for $\ell$-MI, for any $c>1$, would imply  a $\left(2-\frac{1}{c}\right)^{|E|}\cdot \poly(|E|)$  algorithm for $\ell$-MI, contradicting \Cref{thm:3MIMain}.  However, the results of \cite{FGLS19} cannot be applied together with the parameterized algorithm of \cite{huang2023fpt} as its running time is not of the form $c^k\cdot \poly(|E|)$, and cannot be used to rule out a parameterized algorithm with running times such as  $2^{k\log \log k }\cdot \poly(|E|)$.

 We overcome the above hurdles by introducing a generalization of the Monotone Local Search technique. Our generalization converts extension algorithms for implicit set systems, of arbitrary running times, to exponential time algorithms whose running times are better than brute force. 
 We present the results in the context of {\em implicit set problems}, a simple generalization  of the implicit set systems used in \cite{FGLS19} which  allows part of the instance to be only accessible via oracles. 
 The formal statement of our results requires several technical definitions that we give in \Cref{sec:monotone}. Thus, we only provide an informal statement of our main results in this section.
We first consider parameterized algorithms with running time of the form $c^{k^2}$.
 \begin{introInformal}{Informal Statement of \Cref{lem:implicit_ksquare}}
 	Let $\cP$ be an implicit set problem which has an extension algorithm of time   $c^{k^2}$ for some $c\geq 1$. Then there is a randomized $2^{n-\Omega(\log^2 n)} \cdot \poly(n)$ algorithm for $\cP$, where $n=|E|$ is the size of the ground set.
 \end{introInformal}
 
 As the algorithm of Huang and Ward \cite{huang2023fpt} runs in time $c^{k^2} \cdot \poly(n)$, the next result follows from the above and from~\cite{huang2023fpt}, answering positively Question (i). This improves the runtime of the brute force algorithm by a factor of $\Omega\left(2^{\log^2 n}\right) = \Omega\left(n^{\log n}\right)$.   
 
  \begin{restatable}{theorem}{thmalg}
 	\label{thm:faster_MI}
 	For any $\ell \geq 3$, there is an algorithm for \textnormal{oracle $\ell$-matroid intersection} which runs in time $2^{n-\Omega(\log^2(n))}$, where $n$ is the size of the ground set. 
 \end{restatable}

We obtain an analogue to the above result for extension algorithms with runtime of the form $g(k) = 2^{\alpha \cdot k \cdot \log k}$.

\begin{introInformal}{Informal Statement of \Cref{lem:implicit_klogk}}
		Let $\cP$ be an implicit set problem which has an extension algorithm of time   $2^{\alpha k\cdot \log k } $  for some $\alpha > 0$. 
 	Then there is  $2^{n-\Omega\left(n^{\frac{1}{1+\alpha}}\right)} \cdot \poly(n)$ randomized algorithm for $\cP$, where $n$ is the size of the ground set.
\end{introInformal}

Using \Cref{thm:3MIMain} and the above, we conclude that there is no randomized $c^{k \cdot \log k} \cdot \poly(n)$ time algorithm for $3$-MI for some constant $c$. This answers affirmatively Question (ii) and distinguishes $\ell$-MI from the important special case on linear matroids \cite{eiben2024determinantal}. 
 To the best of our knowledge, this gives a new approach for obtaining  parameterized  lower bounds based on exponential-time lower bounds.

\begin{restatable}{theorem}{thmlowerBoundParameterized}
	\label{thm:param_lb}
	For every $\ell \geq 3$ and $\eps>0$ there is no random parameterized algorithm for  \textnormal{oracle $\ell$-matroid intersection} with runtime $\floor{2^{(\ell -2 -\eps) \cdot k\log(k)}}\cdot \poly(n)$, where $n$ is the size of the ground set and $k$ is the rank of the matroids. 
\end{restatable}

Using brute force, one can solve every implicit set problem $\cP$ in time $\approx2^{|E|}$, by enumerating over all subsets of the ground set $E$. We further show that if $\cP$ has an extension algorithm, then it has an algorithm with runtime better than brute force, by a factor which is greater than every polynomial.

 \begin{introInformal}{	
 Informal Statement of \Cref{lem:better_than_brute}}
 Let $\cP$ be an implicit set problem, and assume that $\cP$ has an extension algorithm. Then there is a $2^{n-\omega\left(\log n\right)} \cdot \poly(n)$ algorithm for $\cP$, where $n$ is the size of the ground set.
 \end{introInformal}

\subsection{Matroids Overview}
\label{sec:matroidOverview}

We give below a brief overview of matroid classes and fundamental problems related to our study.

\paragraph{Matroid Classes} Perhaps the simplest class is the one of {\em uniform} matroids. In such matroids, the independent sets $\cI$ are all subsets of the ground set $E$ containing at most $k \in \N$ elements. These matroids are often used to model a {\em cardinality constraint}, commonly studied in various algorithmic settings (e.g., \cite{nemhauser1981maximizing,caprara2000approximation}). 
To model the more general constraints of spanning trees, one needs {\em graphic} matroids \cite{whitney1935t,birkhoff1935abstract}. In these matroids, the ground set $E$ is the set of edges of an undirected graph $G = (V,E)$, and $\cI$ consists of all acyclic subsets of edges.

In this paper we make an extensive use of {\em partition} matroids: given a partition $E_1,\ldots, E_t$ of the ground set $E$ and integer bounds $b_1,\ldots, b_t \in \N$, the set (of independent sets) $\cI$ contains subsets of $E$ consisting of at most $b_i$ elements in $E_i$ for all $i \in \{1,\ldots,t\}$. Finally, {\em linear} matroids \cite{steinitz1913bedingt} generalize all of the above examples: the ground set $E$ represents columns of a matrix, and $\cI$ represents all subsets of linearly independent columns.

We note the existence of {\em non-linear} matroids. One
example is the V\'{a}mos matroid \cite{vamos1968representation}. In fact, asymptotically, the number of linear matroids $-$ compared to the number of non-linear matroids $-$ is negligible \cite{nelson2016almost}. Our main results are based on {\em paving} matroids;
this is the family of all matroids $(E,\cI)$ with rank $r$ (cardinality of a basis) satisfying that every subset $S \subseteq E$ with cardinality $|S| < r$ belongs to $\cI$.
A well-known conjecture of \cite{mayhew2011asymptotic} asserts that asymptotically almost all matroids belong to the class of paving matroids.

\paragraph{Canonical special cases of $3$-MI} A central special case of $3$-MI focuses on linear matroids. As an introduction to the problem, we give below two prominent examples of problems that can be cast as $3$-MI on linear matroids. 
One notable special case of $3$-MI is $3$-{\em dimensional matching} ($3$-DM) \cite{karp1985complexity,karp2010reducibility}. We are given three sets $X_1,X_2,X_3$ of equal cardinality and a collection of triplets $T \subseteq X_1 \times X_2 \times X_3$. A {\em matching} is $M \subseteq T$ such that for any two distinct triplets $(x_1,x_2,x_3),(x'_1,x'_2,x'_3) \in M$ it holds that $x_i \neq x'_i$ for all $i \in \{1,2,3\}$. The goal is to decide if there is a matching of cardinality $|M| = |X_1| = |X_2| = |X_3|$. We give an illustration in \Cref{fig:reductions}. A $3$-DM instance with sets $X_1,X_2,X_3$ and triplets $T$ can be cast as a  $3$-MI instance using three partition matroids. For each $i \in \{1,2,3\}$ define a partition $\left(T^i_a\right)_{a \in X_i}$ of $T$, where for all $a \in X_i$, the set $T^i_a$ consists of all triplets with $a$ in the $i$-th entry; also, the cardinality bound of $T^i_a$ is $1$. Observe that an independent set in the $i$-th matroid can take at most one triplet containing a vertex $a \in X_i$; thus, a common independent set in the three matroids is a matching. Then, the instance contains a matching of cardinality $|X_1|$ if and only if there is a common basis of the three matroids of cardinality $|X_1|$. 

\setlength{\columnsep}{10pt}  

\begin{figure}

	\begin{multicols}{2}  
		\centering
		\begin{tikzpicture}[scale=3]  
			\node[circle, draw, thick, minimum size=1.2cm] (A) at (0,1) {};
			\node[circle, draw, thick, minimum size=1.2cm] (B) at (1,1) {};
			\node[circle, draw, thick, minimum size=1.2cm] (C) at (0,0) {};
			\node[circle, draw, thick, minimum size=1.2cm] (D) at (1,0) {};
			
			\draw[->, thick, red] (B) -- (A);  
			\draw[->, thick, dashed] (D) -- (B);  
			\draw[->, thick, red] (D) -- (C);  
			\draw[->, thick,dashed] (C) -- (A);  
			\draw[->, thick, red] (C) -- (B);  
			\draw[->, thick, black, dashed] (D) -- (A);  
		\end{tikzpicture}
		
		\begin{tikzpicture}[scale=2]  
			\node[circle, draw, thick, minimum size=0.8cm] (A) at (0,2) {};
			\node[circle, draw, thick, minimum size=0.8cm] (B) at (1,2) {};
			\node[circle, draw, thick, minimum size=0.8cm] (C) at (2,2) {};
			\node[circle, draw, thick, minimum size=0.8cm] (D) at (0,1) {};
			\node[circle, draw, thick, minimum size=0.8cm] (E) at (1,1) {};
			\node[circle, draw, thick, minimum size=0.8cm] (F) at (2,1) {};
			\node[circle, draw, thick, minimum size=0.8cm] (G) at (0,0) {};
			\node[circle, draw, thick, minimum size=0.8cm] (H) at (1,0) {};
			\node[circle, draw, thick, minimum size=0.8cm] (I) at (2,0) {};
			\draw[thick, red] (A) -- (B);   
			\draw[thick, red] (B) -- (F);            
			\draw[thick, blue] (D) -- (E);  
			\draw[thick, blue] (E) -- (C);       
			\draw[thick, black] (G) -- (H);       
			\draw[thick, black] (H) -- (I);       
			\draw[thick, green, dashed] (E) -- (I);     
			
			\draw[thick, green, dashed] (A) -- (E); 
			\draw[thick, green, dashed] (E) -- (F); 
			
			\draw[thick, magenta, dashdotted] (D) -- (H);
			\draw[thick, magenta, dashdotted] (H) -- (C);
			
			\draw[thick, orange, dashed] (G) -- (E);   
			
			\draw[thick, orange, dashed] (E) -- (I);   
		\end{tikzpicture}
		
	\end{multicols}

	\caption{\label{fig:reductions}
		Examples of two problems that can be cast as $3$-MI on linear matroids: Hamiltonian Path in directed graphs and $3$-DM.  On the left, a Hamiltonian path in a directed graph is highlighted in (solid line) red. The right figure illustrates a $3$-DM instance where each column of vertices is a dimension, and each triplet characterizes a path of length $2$. An optimal matching is represented by the (solid line) paths in red, blue, and black. 
	}
\end{figure}

Another prominent special case of $3$-MI is Hamiltonian Path on directed graphs.
Given a directed graph $G = (V,E)$, the goal is to decide if there exists a directed Hamiltonian Path in the graph, i.e.,  a directed path that visits each vertex exactly once (see  \Cref{fig:reductions}). An instance of Hamiltonian Path on a directed graph $G = (V,E)$ can be cast as the following $3$-MI instance using two partition matroids and one graphic matroid. For the first partition matroid, define a partition $\left(E_v\right)_{v \in V}$ of the edge set, where $E_v$,  for all $v \in V$, contains all out-edges of $v$. For the second partition define a partition $\left(E'_v\right)_{v \in V}$ of the edge set, where $E'_v$, for all $v \in V$, contains all in-edges of $v$.\footnote{Technically,  we take the {\em truncation} of the above matroids to $|V|-1$, restricting the size of a basis in each of the matroids to be exactly $|V|-1$, which is still a matroid (see, e.g., Chapter 39 in \cite{schrijver2003combinatorial}).} The bound of each set in the partition for both matroids is $1$. This guarantees that a common basis of the two matroids consists of a collection of edges inducing a graph with in and out degree at most $1$ (or, a collection of simple directed paths and cycles). Finally,
the third matroid is a graphic matroid 
defined over the underlying undirected graph of $G$.  Thus, there is a directed Hamiltonian Path in $G$ if and only if there is a common basis of the three matroids. 

While we focus in this paper on the decision version of $\ell$-MI, we note that the maximization version of the problem, in which we seek a maximum cardinality common independent set, has been extensively studied. The current state of the art is $\left(\frac{\ell}{2}+\eps\right)$-approximation for any $\eps>0$ \cite{LeeSV13}, and a recent lower bound which matches this upper bound up to a constant factor \cite{lee2024asymptotically}.

\subsection{Our Techniques}
\label{sec:techniques}

 Our lower bound for $\ell$-MI combines a paving matroid (see \Cref{sec:matroidOverview}) 
with $d = \ell-1$ partition matroids. At a high level, we consider a $d$-dimensional grid as the common ground set of the matroids. 
 The $d$ partition matroids, one for each dimension of the $d$-dimensional grid, enforce the solution to take a specified number of elements $L_{i,j}$, having a specific value $i$ in the $j$-th entry of the grid. This idea is formalized in detail in later sections.

The additional paving matroid is designed to {\em hide} a certain collection of subsets $\cG$ of the grid. Specifically, the bases of this matroid, 
 which are also common bases of the partition matroids, belong to $\cG$. 
  Together with the partition matroids' constraints, finding a common basis for the $\ell$ matroids requires to decide if $\cG \neq \emptyset$. 
In the oracle model, it is a computationally challenging task to decide if $\cG \neq \emptyset$. Thus, essentially, to accomplish this task, an algorithm for $\ell$-MI has to enumerate over all common bases of the $\ell-1$ partition matroids.

The above gives the intuition behind the reduction in \Cref{thm:3MIMain}.
Notably, the $d$-dimensional grid ground set (with $d \geq 2$) plays a crucial role in the construction of the paving matroid. This structure is key to achieving significantly stronger query lower bounds for $3$-MI, compared to the best possible lower bounds known for $2$-matroid intersection~\cite{blikstad2023fast}. 
Hiding a property in the bases of a paving matroid has been used in previous works on matroid problems~\cite{jensen1982complexity,lovasz1980matroid,soto2011simple, berczi2021complexity, doron2024lower, HorschIMOS24}. Nonetheless, these constructions alone are insufficient to obtain a $2^{n - o(n)}$ lower bound for $\ell$-MI. 
In addition, in contrast to the lower bounds derived for a single matroid~\cite{jensen1982complexity,karp1985complexity}, a matroid and a matching constraint \cite{lovasz1980matroid,soto2011simple}, or a matroid and a linear constraint \cite{doron2024lower}, our construction focuses on an {\em arbitrary} number $\ell \geq 3$ of matroids and our bounds become stronger as $\ell$ increases.  

Monotone Local Search \cite{FGLS19} suggests a novel approach for converting an extension algorithm to an exponential time algorithm for an arbitrary  implicit set system. 
Recall that in implicit set systems the input instance encodes a set system $(E,\cF)$, and the objective is to determine if $\cF\neq \emptyset$.
If $\cF\neq \emptyset$ then one can find $S\in \cF$ by guessing $k=|S|$, sampling a random set $X\subseteq \cF$ of size $t \leq k$ using a uniform distribution, and checking if $X$ can be extended to a solution for the problem using the {\em extension algorithm}. If $X\subseteq S$ then the process terminates with `success', and the algorithm finds a set in $\cF$. The above procedure has to be repeated sufficiently many times to ensure a constant success probability.

The running time of monotone local search hinges on the value of $t$. If we select $t=k$, then the algorithm basically guesses a set with hope it is in $\cF$. This leads to a running time similar to that of brute force due to a large number of required repetitions. On the other hand, decreasing the value of $t$ reduces the number of required repetitions, while increasing the running time of the extension algorithm used in each repetition. A main observation in the analysis of monotone local search is that the optimal value of $t$ is bounded away from $k$, implying that the resulting running time is better than brute force. This property was previously known  for the case in which the running time of the extension algorithm was of the form $f(k)=c^k$ times a polynomial factor \cite{FGLS19}. We show that, interestingly,
the same property holds for arbitrary $f$.

	\paragraph{Organization} \Cref{sec:prel} gives some definitions and notation. In \Cref{sec:Fmatroids} we construct a family of matroids that will be used to prove \Cref{thm:3MIMain}. We give the proof of \Cref{thm:3MIMain} in \Cref{sec:3MI}, respectively. \Cref{sec:monotone} describes a generalization of the monotone local search technique of \cite{FGLS19} and gives the proofs of our remaining results. We conclude with a discussion and open problems in \Cref{sec:discussion}. 
	For convenience, the first page of the paper contains a table of contents.

\section{Preliminaries}
\label{sec:prel}

\paragraph{Notations}
We use $\N = \{1,2,\ldots\}$ to denote the set of natural numbers, excluding zero.
For any $n \in \N$, let $[n] = \{1,\ldots, n\}$ for short. Let $|I|$ be the encoding size of instance $I$ of a decision problem $\cD$.  For some $n \in \N$, a vector $\vb \in \N^n$, and an entry $i \in [n]$, let $\vb_i$ be the $i$-th entry in $\vb$. Similarly, for some $m,n \in [n]$, a matrix $A \in \N^{m \times n}$,  $i \in [m]$, and $j \in [n]$, let $A_{i,j}$ denote the entry of the matrix in row $i$ and column $j$. 
For a set $A$ and some element $e$ let $A+e$ and $A-e$ denote $A \cup \{e\}$ and $A \setminus \{e\}$, respectively. We use $\log = \log_2$ to denote the logarithm in base $2$ and $\poly(n)$ to denote polynomial functions of a variable $n$.  Let $X,Y \subseteq \N$ and let $\pi:X \rightarrow Y$ be some bijection. For every $S \subseteq X$, let $\pi(S) = \{\pi(i) \mid i \in S\}$ and for every $T \subseteq Y$ let $\pi^{-1}(T) = \{i \in X \mid \pi(i) \in T\}$.

\paragraph{Matroids}
\label{sec:matroids}
\comment{A {\em matroid} is a set system $(E, \cI)$, where $E$ is a finite set and $\cI \subseteq 2^E$ are the {\em independent sets} such that $(i)$ $\emptyset \in \cI$, $(ii)$ for all $A \in \cI$ and $B \subseteq A$ it holds that $B \in \cI$, and $(iii)$ for all $A,B \in \cI$ where $|A| > |B|$, there is $e \in A \setminus B$ such that $B + e \in \cI$. Property $(ii)$ is called the {\em hereditary property},
	and Property $(iii)$ is called the  {\em exchange property}. }
	
	We repeat some of the definitions given in the introduction more rigorously. The {\em rank} of a matroid $(E,\cI)$ is the maximum cardinality of an independent set: $\max_{S \in \cI} |S|$ and a {\em basis} is an inclusion-wise maximal independent set - of cardinality equals to the rank (see, e.g., Chapter 39 in \cite{schrijver2003combinatorial} for more details). A matroid $(E,\cI)$ with rank $r$ is a {\em paving matroid} if for each $S \subseteq E$ such that $|S| < r$ it holds that $S \in \cI$. 
	For some matroid $(E,\cI)$, let $\textnormal{\textsf{bases}}(\cI) = \{S \in \cI \mid |S| = \max_{S' \in \cI} |S'|\}$ be the set of bases of $\cI$. For some matroids $(E,\cI_1),\ldots, (E,\cI_d)$, we use $\textnormal{\textsf{bases}}(\cI_1,\ldots,\cI_d)$ to denote the the intersection of the sets of bases  $ \textnormal{\textsf{bases}}(\cI_1) \cap \ldots \cap \textnormal{\textsf{bases}}(\cI_d)$.

Given a matroid $M = (E,\cI)$ and $S \subseteq E$ let $\cI_{\cap S} = \{T \subseteq S \mid T \in \cI\}$ and let $M_{\cap S} = (S, \cI_{\cap S})$ be the {\em restriction} of $M$ to $S$. 
Also, Given a matroid $M=(E,\cI)$ and $S\in \cI$ let $\cI/S = \{T\subseteq E\setminus S\,|\,T\cup S \in \cI\}$ and let $M/S= (E\setminus S, \cI/S)$ be the {\em contraction} of $M$ by $S$. 
It is well known that $M_{\cap S}$ and $M/S$ are matroids (see, e.g., \cite{schrijver2003combinatorial}).

In terms of matroid representation, note that $\ell$-MI (for some $\ell \geq 3$) is an abstract problem that do not specify the input. In an instance $(E,\cI_1,\ldots,\cI_{\ell})$ of {\em oracle}-$\ell$-MI, the ground set $E$ is explicitly given in the input, while the set $\cI_i$, for $i \in \{1,\ldots,\ell\}$, can be accessed only via a designated membership oracle, which determines if some set $S \subseteq E$ belongs to $\cI_i$ in a single query.

\paragraph{Randomized Algorithms}

An instance $I$ of a decision problem $\cD$
is a ``yes"-instance if the correct answer for $I$ is ``yes"; otherwise, $I$ is a ``no"-instance.
We say that $\cA$ is a {\em randomized algorithm} for a decision problem $\cD$ if,
given a ``yes"-instance $I$ of $\cD$, $\cA$ returns ``yes"  with probability at least $\frac{1}{2}$; for a ``no"-instance, $\cA$ returns ``no" with probability $1$.

	\paragraph{The Empty Set Problem}
We give a reduction from the following problem. Intuitively, in this problem we have two players Alice and Bob.\footnote{We remark that this problem is described using the two players merely to provide more intuition for the reader and is not necessary for the formal definition of the problem.} Alice receives a set of numbers $[n] = \{1,\ldots,n\}$ and needs to determine if a set $\cF \subseteq 2^{[n]}$, which is not explicitly given to her, is empty or not. Alice gains information about $\cF$ only by querying Bob - an oracle. Namely, for each set $S \subseteq [n]$, Bob replies either $\textnormal{\textsf{true}}$ - implying that $S \in \cF$, or $\textnormal{\textsf{false}}$ otherwise.  More concretely, for any $n,k \in \N \cup \{0\}$ let $$\cS_{n,k} = \left\{ S \subseteq [n] ~\big|~ |S| = k \right\}.$$ 

Define the Empty Set problem as follows. 

 \problem{\label{def:emptySet}{ 
		Empty Set (ES)}}{
	{\bf Instance} &  
	$(n,k,\cF)$, where $n,k \in \N$ and $\cF \subseteq \cS_{n,k}$.\\
	{\bf Objective} & Decide if $\cF \neq \emptyset$. 
}

We give an illustration in \Cref{fig:ES}. 
\begin{figure}	
	\begin{center}
		\begin{tikzpicture}[ultra thick,scale=1.1, every node/.style={scale=1}]			
			\node at (-0.8,3.5) {$\bf \textcolor{black}{\cS_{4,2}=} $};
			\draw (0,3) rectangle (6,4);
			\draw (1,3)--(1,4);
			\draw (2,3)--(2,4);
			\draw (3,3)--(3,4);
				\draw (4,3)--(4,4);
			\draw (5,3)--(5,4);
			\draw (6,3)--(6,4);								
			\node at (0.5,3.5) {$\{1,2\}$};
			\node at (1.5,3.5) {$\{1,3\}$};
			\node at (2.5,3.5) {$\{2,3\}$};
			\node at (4.5,2.7) {$\in \cF$};
			\node at (3.5,3.5) {$\{1,4\}$};		
				\node at (4.5,3.5) {$\bf \textcolor{red}{\{2,4\}}$};		
					\node at (5.5,3.5) {$\{3,4\}$};		
		\end{tikzpicture}
	\end{center}
	\caption{\label{fig:ES} An example of an Empty Set ``yes''-instance with universe $[n] = \{1,2,3,4\}$, $k = 2$, and $\cF = \{\{2,4\}\}$ (highlighted in red). }
\end{figure}
Note that an instance $(n,k,\cF)$ of the Empty Set is a ``yes''-instance if and only if $\cF \neq \emptyset$; we refer to $n$ as the {\em size} of the universe and to $k$ as the {\em cardinality target}.  Similarly to 3-MI, the Empty Set problem is an abstract problem that does not specify the input; clearly, if $\cF$ is given in the input, the problem becomes trivial. We consider an oracle model for this problem. Specifically, in an instance $(n,k,\cF)$ of {\em oracle}-ES, the numbers $n,k$ are explicitly given in the input, while the set $\cF$ can be accessed only via a membership oracle, which indicates whether some $S \in \cS_{n,k}$ belongs to $\cF$ in a single query.

We obtain the following lower bound on the minimum number of queries needed to decide the oracle model of the Empty Set problem. The proof is given in \Cref{sec:proofs}. 

\begin{restatable}{lemma}{lemES}
	\label{lem:EmptySet}
	For every $n \in \N$ and $k \in \{0,1\ldots,n\}$, there is no randomized algorithm that decides \textnormal{oracle-ES} problem on instances with a universe of size $n$ and cardinality target $k$ in fewer than $\frac{{n \choose k}}{2}$ oracle queries.
\end{restatable}

\section{$\cG$-Matroids}
\label{sec:Fmatroids}

 In this section we introduce the class of {\em $\cG$-matroids} that will be used to prove \Cref{thm:3MIMain}. On a high level, a $\cG$-matroid is a matroid whose ground set is the $d$-dimensional grid denoted as $\gr = [n]^d$, for some $n,d \in \N$ (simply $[n] \times [n]$ for $d = 2$).
  The independent sets of the matroid are defined according to a matrix $L \in \N^{n \times d}$. Each entry $L_{i,j}$ of $L$ indicates a target value for the number of entries with value $i$ in dimension $j$ of $\gr$, for $1 \leq i \leq n$, $1 \leq j \leq d$.
 {The bases of the matroid are defined as those which violate one of the target values, or adhere to the targets and satisfy an additional property. } Subsets of $\gr$ satisfying all targets with equality are called $L$-{\em perfect} (see~\Cref{fig:perfect}).\footnote{The formal definition of the independent sets of a $\cG$-matroid is given below).} 
 	For example, $L_{2,1} = 3$ defines a target value of $3$ on the number of elements in $\gr$ whose value in the first dimension is equal to $2$. When $d=2$, the first dimension is the row number, and the second is the column number. Then $L$ defines $n \cdot 2$ targets, one for each row and one for each column of $\gr = [n] \times [n]$, and
 	$L_{2,1} = 3$ implies that in any  $L$-prefect set there are exactly $3$ elements from row $2$. 

 	 \begin{figure}[h]    
 		\begin{center}
 			\begin{tikzpicture}[ultra thick,scale=1.1, every node/.style={scale=1}]            
 				
 				\draw (0,5) rectangle (6,6);
 				\foreach \x in {1,2,3,4,5,6} {
 					\draw (\x,5)--(\x,6);
 				}
 				
 				\node at (0.5,5.5) {$*$};
 				\node at (1.5,5.5) {};
 				\node at (2.5,5.5) {$*$};
 				\node at (3.5,5.5) {};
 				\node at (4.5,5.5) {$*$};
 				\node at (5.5,5.5) {};
 				
 				\draw (0,4) rectangle (6,5);
 				\foreach \x in {1,2,3,4,5,6} {
 					\draw (\x,4)--(\x,5);
 				}
 				\node at (0.5,4.5) {};
 				\node at (1.5,4.5) {$*$};
 				\node at (2.5,4.5) {};
 				\node at (3.5,4.5) {$*$};
 				\node at (4.5,4.5) {};
 				\node at (5.5,4.5) {$*$};
 				
 				\draw (0,3) rectangle (6,4);
 				\foreach \x in {1,2,3,4,5,6} {
 					\draw (\x,3)--(\x,4);
 				}
 				\node at (0.5,3.5) {$*$};
 				\node at (1.5,3.5) {};
 				\node at (2.5,3.5) {};
 				\node at (3.5,3.5) {$*$};
 				\node at (4.5,3.5) {};
 				\node at (5.5,3.5) {$*$};
 				
 				\draw (0,2) rectangle (6,3);
 				\foreach \x in {1,2,3,4,5,6} {
 					\draw (\x,2)--(\x,3);
 				}
 				\node at (0.5,2.5) {};
 				\node at (1.5,2.5) {$*$};
 				\node at (2.5,2.5) {$*$};
 				\node at (3.5,2.5) {};
 				\node at (4.5,2.5) {$*$};
 				\node at (5.5,2.5) {};
 				
 				\draw (0,1) rectangle (6,2);
 				\foreach \x in {1,2,3,4,5,6} {
 					\draw (\x,1)--(\x,2);
 				}
 				\node at (0.5,1.5) {$*$};
 				\node at (1.5,1.5) {};
 				\node at (2.5,1.5) {$*$};
 				\node at (3.5,1.5) {};
 				\node at (4.5,1.5) {};
 				\node at (5.5,1.5) {$*$};
 				
 				\draw (0,0) rectangle (6,1);
 				\foreach \x in {1,2,3,4,5,6} {
 					\draw (\x,0)--(\x,1);
 				}
 				\node at (0.5,0.5) {};
 				\node at (1.5,0.5) {$*$};
 				\node at (2.5,0.5) {};
 				\node at (3.5,0.5) {$*$};
 				\node at (4.5,0.5) {$\bf \textcolor{black}{*}$};
 				\node at (5.5,0.5) {};
 				
 			\end{tikzpicture}
 		\end{center}
 		\caption{\label{fig:perfect} An example of an $L$-perfect set. The figure shows a $6 \times 6$ grid, i.e., $n = 6$ and $d = 2$ implying $\gr = [6]^2$. In the matrix $L \in \N^{6 \times 2}$ all entries are equal to $3$. The entries of one $L$-perfect set are marked with `*' (for example, the $L$-perfect set contains the elements $(1,1)$, $(1,3)$ and $(1,5)$ from the first row).}
 	\end{figure}
 	
 	We can choose $n$ and $L$ such that the number of $L$-perfect subsets of $\gr$ is extremely large and asymptotically close to the total number of subsets of $\gr$. Hence, intuitively, finding one specific $L$-perfect set (based on querying an oracle) is as hard as finding a needle in a haystack. Specifically, fix some arbitrary collection of subsets $\cG \subseteq 2^{\gr}$. The set of bases of the $\cG$-matroid consists of all subsets of $\gr$ whose cardinality is equal to the sum of the first column of $L$, that are either (i) not $L$-perfect, or (ii) are $L$-perfect and belong to $\cG$ (see \Cref{def:Fmatroid}). This is illustrated in \Cref{fig:Fmatroids}. The above suggests that, roughly, finding an $L$-perfect set in $\cG$ is as hard as solving the Empty Set problem. Before we show this in detail, let us make the definition of $\cG$-matroid more precise.

 	\begin{figure}[h]    
 		\begin{center}
 			\begin{tikzpicture}[ultra thick,scale=1.1, every node/.style={scale=1}]            
 				\draw (0,2) rectangle (2,3);
 				\foreach \x in {1,2} {
 					\draw (\x,2)--(\x,3);
 				}
 				
 				\node at (0.5,2.5) {$*$};
 				\node at (1.5,2.5) {};
 				
 				\draw (0,1) rectangle (2,2);
 				\foreach \x in {1,2} {
 					\draw (\x,1)--(\x,2);
 				}
 				\node at (0.5,1.5) {};
 					\node at (-0.7,2) {$\gr = $};
 				\node at (1.5,1.5) {$*$};
 				
 				\node at (1,0.5) {$\textnormal{\textsf{bases}} = \{\bf \textcolor{red}{\{(1,1),(2,2)\}},\{(1,1),(1,2)\},\{(2,1),(2,2)\},\{(1,1),(2,1)\} \{(1,2),(2,2)\} \}$};
 			\end{tikzpicture}
 		\end{center}
 			\caption{\label{fig:Fmatroids} An illustration of the bases of a $\cG$-matroid with the grid defined as $\gr = [n]^d$ for $n,d = 2$, and in the target matrix $L$ all $n \cdot d =4$ entries are equal to $1$. Suppose that $\cG$ contains a single set, whose entries are marked with '*', i.e., $\cG = \{\{(1,1), (2,2)\}\}$. The bases of this $\cG$-matroid consist of $\{(1,1), (2,2)\}$ (in red) $-$ which is both in $\cG$ and $L$-perfect, and all non $L$-perfect subsets of cardinality $2 = L_{1,1}+L_{2,1}$,the sum of entries of the first column in $L$.
 			}
 	\end{figure}

 	A member in the $\cG$-matroid class is characterized by four parameters: $n,d,L$, and $\cG$. The first two parameter, $n$ and $d$, describe the ground set; namely, for some $n,d \in \N$, where $n,d \geq 2$, let $\gr_{n,d} = [n]^d$ be the {\em grid} of $n,d$, used as the ground set of the $\cG$-matroid. When clear from the context, we simply use $\gr = \gr_{n,d}$.  The third parameter is the target matrix $L \in \N^{n \times d}$. We assume that $L$ is simple-uniform (this is required to ensure that $\cG$-matroids  are indeed matroids). 

	\begin{definition}
	\label{def:SU'}
	For some $n,d \in \N$, where $n,d \geq 2$, a matrix $L \in \N^{n \times d}$ is {\em simple-uniform (SU)} if the following holds.
	\begin{enumerate}
		\item For all $i \in [n]$ and $j \in [d]$ it holds that $L_{i,j} \in \{0,1\ldots, n\}$. 
		\item For all $j \in [d]$ it holds that $2 \leq \sum_{i \in [n]} L_{i,j} \leq n^d-1$.
		\item For all $j_1,j_2 \in [d]$ it holds that $\sum_{i \in [n]} L_{i,j_1} = \sum_{i \in [n]} L_{i,j_2}$.
	\end{enumerate}
\end{definition}

 An SU matrix $L$ is associated with a collection of $L$-{\em perfect} subsets $S$ of $\gr$ which take exactly $L_{i,j}$ elements whose $j$-th entry equals to $i$. This notion of $L$-perfectness is used in the definition of $\cG$-matroid.  
 \begin{definition}
 	\label{def:perfect}
 	 Let $n,d \in \N$, where $n,d \geq 2$, and let 
 	$L \in \N^{n \times d}$ be an SU matrix.  We say that $S \subseteq \gr$ is $L$-{\em perfect} if for all $i \in [n]$ and $j \in [d]$ it holds that $\left|\left\{ \ve \in S \mid \ve_j = i\right\} \right| = L_{i,j}$.
 \end{definition}

 The last parameter describes a class of subsets in the grid: $\cG \subseteq 2^{\gr}$. As explained above, we define the set of bases of the $\cG$-{\em matroid} as all subsets of $\gr$ whose cardinality is equal to the sum of the first column of $L$ that are either (i) not $L$-perfect, or (ii) $L$-perfect and belong to $\cG$. Formally, for some ground set $E$ and $\cB \subseteq 2^E$, let $\cI(\cB) = \left\{ S \subseteq E \mid \exists B \in \cB : S \subseteq B \right\}$ be the {\em independent sets} of $\cB$. Define the $\cG$-matroid class as follows.

 \begin{definition}
 	\label{def:Fmatroid}
 Let $n,d \in \N$, where $n,d \geq 2$, 
 $L \in \N^{n \times d}$ an SU matrix, and $\cG \subseteq \gr$. 
 	Define the {\em bases} of $n,d,L$ and $\cG$ as
 	$\cX \subseteq 2^{\gr}$, where $S \subseteq \gr$ belongs to $\cX$ if and only if $|S| = \sum_{i \in [n]} L_{i,1}$ and 
 	one of the following holds. 
 	
 	\begin{itemize}
\item $S$ is {\bf not} $L$-perfect. 
\item $S$ is $L$-perfect {\bf and} $S \in \cG$. 
 	\end{itemize}
 	
 	Define the $\cG$-{\em matroid} of $n,d$ and $L$ as $M = (\gr,\cI(\cX))$. 
 \end{definition}

 Observe that $\cG$ can be an {\em arbitrary} set. This property is crucial for our reductions. Furthermore, note that $S$ could have more than $L_{i,j}$ elements whose $j$-th entry equals $i$.
To prove that $\cG$-matroids are indeed matroids, we use the following result of \cite{frank2011connections}, stated with our notations.

\begin{lemma}
	\label{lem:frank}
	\textnormal{[Theorem 5.3.5 in \cite{frank2011connections}]} Let $k \geq 2$ be an integer and $E$ a set of size at least $k$. Let $\cH = \{H_1,\ldots, H_t\}$ be a (possibly empty) set-system of proper subsets of $E$ in which every set $H_i$ has at least $k$ elements, and the intersection of any two of them has at most $k-2$ elements. Let 
	$$\cB_{\cH} = \left\{  B \subseteq E \Xspace \big|\Xspace |B| = k, B \not \subseteq H_i \Xspace \forall i \in [t] \right\}.$$
	Then, $\left(E,\cI(\cB_{\cH})\right)$ is a paving matroid. 
\end{lemma}

Using \Cref{lem:frank}, we show that $\cG$-matroids are a subclass of paving matroids. 

\begin{lemma}
	\label{lem:mat}
Let $n,d \in \N$, where $n,d \geq 2$, $L \in \N^{n \times d}$ an \textnormal{SU} matrix, and $\cG \subseteq 2^{\gr}$. Then, the $\cG$-matroid of $n,d,L$ is a paving matroid. 
\end{lemma}

\begin{proof}
	Let 
	$$\cH = \left\{ H \in 2^{\gr} \setminus \cG \Xspace \Big| \Xspace H \textnormal{ is $L$-perfect}\right\}$$ 
	be the collection of $L$-perfect subsets of $\gr$ that are not in $\cG$.
	We show that $\cH$ satisfies the properties of \Cref{lem:frank}. 
Let $k = \sum_{i \in [n]} L_{i,1}$. Observe that $2 \leq k \leq n^d-1$, since $L$ is an SU matrix.  We use several auxiliary claims.

	\begin{claim}
	\label{claim:Hprop1}
	For all $H \in \cH$ it holds that $|H| = k$ and $|H| < |\gr|$.   
\end{claim}

\begin{claimproof}
	As every $H \in \cH$ is $L$-perfect, we have that
	\begin{equation}
		\label{eq:proper}
		|H| = \sum_{i \in [n]} \left|\left\{ \ve \in H \mid \ve_1 = i\right\} \right| = \sum_{i \in [n]} L_{i,1} = k \leq n^d-1 < |\gr|. 
	\end{equation} 
	The first inequality holds since $L$ is an SU matrix. 
\end{claimproof}
	
		\begin{claim}
		\label{claim:Hprop2}
		For all $H_1,H_2 \in \cH$ it holds that $|H_1 \cap H_2| \leq k-2$.   
	\end{claim}
	
	\begin{claimproof}
		Let $H_1, H_2 \in \cH$ such that $H_1 \neq H_2$. Since $|H_1| = |H_2| = k$, 
	there is $\ve^1 \in H_1 \setminus H_2$ and $\ve^2 \in H_2 \setminus H_1$. 
	Hence, there is $j \in [d]$ such that $\ve^1_j \neq \ve^2_j$, where recall that $\ve^1_j, \ve^2_j$ are the $j$-th entries of $\ve^1$ and $\ve^2$, respectively. Let $i = \ve^1_j$.
	As $H_1,H_2 \in \cH$, $H_1$ and $H_2$ are $L$-perfect we have,
	\begin{equation}
		\label{eq:e1}
		\left|\left\{ \ve \in H_1 \mid \ve_j = i\right\} \right| =  L_{i,j} = \left|\left\{ \ve \in H_2 \mid \ve_j = i\right\} \right| .
	\end{equation} 
	
	By \eqref{eq:e1} and since $\ve^1 \in \left\{ \ve \in H_1 \mid \ve_j = i\right\} \setminus \left\{ \ve \in H_2 \mid \ve_j = i\right\}$, there is $\ve^{*} \in \left\{ \ve \in H_2 \mid \ve_j = i\right\}$ such that $\ve^* \notin \left\{ \ve \in H_1 \mid \ve_j = i\right\}$. Observe that $\ve^{*}_j = i = \ve^1_j \neq \ve^2_j$; thus, $\ve^{*} \neq \ve^2$. Finally, we note that $\ve^{*},\ve^2 \in H_2 \setminus H_1$. Hence, as $|H_1| = |H_2| = k$, it follows that 
	$|H_1 \cap H_2| \leq k-2$. 
	\end{claimproof}

	By \Cref{claim:Hprop1,claim:Hprop2} we have that $\cH$ satisfies the properties of \Cref{lem:frank}. Define
	$$\cB_{\cH} = \left\{  B \subseteq \gr \Xspace \big|\Xspace |B| = k, B \not \subseteq H \Xspace \forall H \in \cH \right\}.$$
	 Then,
	by \Cref{lem:frank}, $(\gr,\cI(\cB_{\cH}))$ is a paving matroid.
	Let $\cX$ be the bases of $n,d,L$ and $\cG$ (see \Cref{def:Fmatroid}). Observe that, by \Cref{claim:Hprop1}, $|H| = k$ for all $H \in \cH$; thus, for all $B \subseteq \gr$ it holds that $B \in \cB_{\cH}$ if and only if $|B| = k$ and $B \notin \cH$. We conclude with the following auxiliary claim.
	
		\begin{claim}
		\label{claim:X=Bh}
		$\cB_{\cH} = \cX$.
	\end{claim}
	
	\begin{claimproof}
		Let $B \in \cB_{\cH}$. Then, $|B| = k$ and $B \notin \cH$; therefore, $|B| = k$ and either (i) $B$ is not $L$-perfect, or (ii) $B$ is $L$-perfect and $B \in \cG$. Hence, by \Cref{def:Fmatroid} $B \in \cX$. For the second direction, let $S \in \cX$. Then, $|S| = k$, and either (i) $S$ is not $L$-perfect or (ii) $S$ is $L$-perfect and $S \in \cG$. Therefore, $|S| = k$ and $S \notin \cH$, implying that $S \in \cB_{\cH}$. We conclude that 
		$\cB_{\cH} = \cX$. 
	\end{claimproof}

	By \Cref{claim:X=Bh} we have that $\cI(\cX) = \cI(\cB_{\cH})$. Consequently,
	the $\cG$-matroid $(\gr,\cI(\cX))$ is indeed a paving matroid by \Cref{lem:frank}.     
\end{proof}

  Let $n,d \in \N$, where $n,d \geq 2$, and let $L \in \N^{n \times d}$ be an SU matrix. Observe that $L$-perfect subsets can be cast as the intersection of bases of $d$ partition matroids.
  Specifically, for all $j \in [d]$, let
  \begin{equation}
  	\label{eq:constraints}
  	\cI_{L,j} = \left\{ S \subseteq \gr \Xspace\Big|\Xspace \left|\left\{ \ve \in S \mid \ve_j = i\right\} \right|  \leq L_{i,j} \Xspace\forall i \in [n]\right\},
  \end{equation}
  and let $(\gr,\cI_{L,j})$ be the $(L,j)$-{\em partition matroid} of $\gr$. That is, we define 
  one matroid for each entry $j \in [d]$, and the partition of $\gr $ corresponding to this $j$-th matroid is induced by all values $\{1,\ldots,n\}$ in the $j$-th entry (of elements in $\gr$). 
 The following is a fundamental property of partition matroids. We give the proof for completeness in \Cref{sec:proofs}. 
 
 	\begin{restatable}{lemma}{partitionBases}
 	\label{claim:partitionBases}
 	Let $n,d \in \N$, where $n,d \geq 2$, and let $L \in \N^{n \times d}$ be an SU matrix. Then, for all $j \in [d]$ it holds that 
 	\begin{equation*}
 		\label{eq:basesLj}
 		\textnormal{\textsf{bases}}(\cI_{L,j}) = \left\{ B \subseteq \gr \Xspace\Big|\Xspace \left|\left\{ \ve \in B \mid \ve_j = i\right\} \right|  = L_{i,j} \Xspace\forall i \in [n]\right\}.
 	\end{equation*}
 \end{restatable}

 We give below some core properties of the $\cG$-matroid class. The next lemma follows directly from the definition of $L$-perfect subsets.

 \begin{lemma}
 	\label{obs:L}
 	Let $n,d \in \N$, where $n,d \geq 2$, and let $L \in \N^{n \times d}$ be an SU matrix. Then, $S \subseteq \gr$ is $L$-perfect if and only if  $S \in \textnormal{\textsf{bases}}(\cI_{L,1},\ldots, \cI_{L,d})$.  
 \end{lemma}
 
 \begin{proof}

 	By \Cref{claim:partitionBases} and~\Cref{def:perfect}, $S$ is $L$-perfect if and only if $S \in \textnormal{\textsf{bases}}(\cI_{L,1},\ldots, \cI_{L,d})$.  
 \end{proof}
 
 Next, we show that common bases of all $(L,j)$-partition matroids intersected with an $\cG$-matroid must be in $\cG$ as well (for every selection of $L$ and $\cG$). 
 
 \begin{lemma}
 	\label{lem:properties}
 	Let $n,d \in \N$, where $n,d \geq 2$, and let $L \in \N^{n \times d}$ be an SU matrix. Also, let $\cG \subseteq 2^{\gr}$, and denote by $\cX$ the set of bases of $n,d,L$, and $\cG$. Then, for all $B \in \cX$ such that  $B \in \textnormal{\textsf{bases}}(\cI_{L,1},\ldots, \cI_{L,d})$, it holds that $B \in \cG$. 
 \end{lemma}
 
 \begin{proof}
 	 As $B \in \textnormal{\textsf{bases}}(\cI_{L,1},\ldots, \cI_{L,d})$, by \Cref{obs:L} $B$ is $L$-perfect. Since $B \in \cX$ and $B$ is $L$-perfect, by \Cref{def:Fmatroid}  $B \in \cG$.
 \end{proof}
 
 The next result will also be useful for our reductions. 
  \begin{lemma}
 	\label{lem:properties2}
 	Let $n,d \in  \N$ such that $n,d \geq 2$, and let $\cG \subseteq 2^{\gr}$. Also, let 
 	$S \in \cG$ such that $2 \leq |S| \leq |\gr|-1$. Then, there is an \textnormal{SU} matrix $L \in \N^{n \times d}$ such that $S \in \textnormal{\textsf{bases}}(\cI(\cX), \cI_{L,1},\ldots, \cI_{L,d})$, where $\cX$ are the bases of $n,d,L$ and $\cG$. 
 	 
 \end{lemma}
 
 \begin{proof}
 	For all $j \in [d]$  and $i \in [n]$ define $L_{i,j} = \left|\left\{ \ve \in S \mid \ve_j = i\right\} \right|$.  
 	As for all $j \in [d]$ it holds that $\left(\left\{ \ve \in S \mid \ve_j = i\right\} \right)_{i \in [n]}$ is a partition of $S$, for all $j \in [d]$ we have that
 	\begin{equation}
 		\label{eq:Lij}
 		\sum_{i \in [n]} L_{i,j} = \sum_{i \in [n]} \left|\left\{ \ve \in S \mid \ve_j = i\right\} \right| = |S|.
 	\end{equation} 
 	
 	As $2 \leq |S| \leq |\gr|-1$, we have that $2 \leq \sum_{i \in [n]} L_{i,1} \leq |\gr|-1 = n^d-1$. Moreover, by \eqref{eq:Lij}, for all $j_1,j_2 \in [d]$ it holds that 
 	$\sum_{i \in [n]} L_{i,j_1} = |S| = \sum_{i \in [n]} L_{i,j_2}$. 
 	We conclude that $L$ is an SU matrix. 
 	For all $j \in [d]$ it holds that $\left|\left\{ \ve \in S \mid \ve_j = i\right\} \right|  = L_{i,j}$; thus, $S$ is $L$-perfect. By \Cref{obs:L} $S \in \textnormal{\textsf{bases}}(\cI_{L,1},\ldots, \cI_{L,d})$. Finally, since $S$ is $L$-perfect and $S \in \cG$, we have that $S \in \cI(\cX)$ by \Cref{def:Fmatroid}. Hence, $S \in \textnormal{\textsf{bases}}(\cI(\cX),\cI_{L,1},\ldots, \cI_{L,d})$ as required.  
 \end{proof}
 
 Next, we show that we can easily obtain a membership oracle for $\cG$-matroids given a membership oracle for $\cG$ itself.

 \begin{lemma}
 	\label{obs:oracle}
 	Let $n,d \in  \N$ such that $n,d \geq 2$, $L \in \N^{n \times d}$ an SU matrix, and $\cG \subseteq 2^{\gr}$. Given a membership oracle $H$ for $\cG$, we can construct a membership oracle $T$ for the $\cG$-matroid of $n,d$ and $L$ that satisfies the following properties. 
 	\begin{enumerate}
 		\item 
 		Each query to $T$ requires a single query to $H$.
 		\item $T$ performs $O\left(n^d\right)$ operations per query.  
 	\end{enumerate}
 \end{lemma}
 
 \begin{proof}
 	Let $\cX$ be the bases of $n,d,L$ and $\cG$. Define a membership oracle $T$ for the $\cG$-matroid of $n,d$ and $L$ such that for any $S \subseteq \gr$ the oracle $T$ answers as follows.
 	\begin{itemize}
 		\item If $|S| < \sum_{i \in [n]} L_{i,1}$ then $T$ returns that $S \in \cI(\cX)$.
 		
 		\item If $|S| > \sum_{i \in [n]} L_{i,1}$ then $T$ returns that $S \notin \cI(\cX)$.
 		
 		\item If $|S| = \sum_{i \in [n]} L_{i,1}$ then $T$ checks if $S$ is $L$-perfect, queries the oracle $H$, and answers that $S \in \cI(\cX)$ if and only if one of the following conditions hold.
 		\begin{itemize}
 			\item $S$ is not $L$-perfect.
 			\item $S$ is $L$-perfect and $H$ returns that $S \in \cG$.
 		\end{itemize}
 	\end{itemize} Clearly, each query of $T$ requires at most one query to $H$. In addition, for all $S \subseteq \gr$ determining if $S$ is $L$-perfect can be done in time $O \left(n^d\right)$. Hence, $T$ satisfies the query and running time complexity as stated in the lemma. 
 	
 	It remains to prove correctness. Since $(\gr, \cI(\cX))$ is a paving matroid (by \Cref{lem:mat}), for all $S \subseteq \gr$ such that $|S| < \sum_{i \in [n]} L_{i,1}$ it holds that $S \in \cI(\cX)$. Moreover, by \Cref{def:Fmatroid}, for all $S \subseteq \gr$ such that $|S| > \sum_{i \in [n]} L_{i,1}$ it holds that $S \notin \cI(\cX)$. Finally, by \Cref{def:Fmatroid}, for all $S \subseteq \gr$ such that $|S| = \sum_{i \in [n]} L_{i,1}$, since $H$ is a membership oracle for $\cG$, $T$ correctly decides  if $S \in \cI(\cX)$. This gives the statement of the lemma. 
 \end{proof}

\section{A Lower Bound for Oracle $\ell$-Matroid Intersection}
\label{sec:3MI}
		
In this section we establish~\Cref{thm:3MIMain}. The proof is based on a reduction from the Empty Set problem. Specifically, given an Empty Set (ES) instance $I = (n,k,\cF)$,\footnote{ See \Cref{sec:prel} for the definition of the ES problem.} we construct a collection of {\em reduced} $\ell$-MI instances, for some fixed $\ell \in \N$. All of the reduced instances have the same ground set $\gr = \left[ \sqrt[d]{n} \,\right]^{d}$,
where $d = \ell-1$, implying that $|\gr| = n$. In addition, all reduced instances are defined with respect to a projection $\cG$ of $\cF$ to $\gr$. For each simple-uniform (SU) matrix $L$ with dimensions $n \times d$, we generate one reduced $\ell$-MI instance $R_L(I)$. The corresponding matroids of $R_L(I)$ are the $(L,j)$-partition matroid of $\gr$, for all $j \in [d]$, and the $\cG$-matroid of $n,d$ and $L$. More formally,

\begin{definition}
	\label{def:Reduced3MI}
	Let $\ell \geq 3$, $d = \ell-1$, and $n \in \N$ such that $n^{\frac{1}{d}} \in \N$, where $n^{\frac{1}{d}} \geq 2$. Also, let $\gr = \left[ \sqrt[d]{n} \,\right]^{d}$, and $\pi:[n] \rightarrow \gr$ be an arbitrary fixed bijection.
	Given an Empty Set instance $I = (n,k,\cF)$, 
	for any \textnormal{SU} matrix $L \in \N^{N \times d}$ define the {\em reduced} \textnormal{$\ell$-MI} instance of $L$ and $I$, $R_L(I) = (\gr,\cI^{\cG}_L, \cI_{L,1},\ldots, \cI_{L,d})$, as follows. 
	\begin{itemize}
		
		\item  Define $\cG = \{\pi(S) \mid S \in \cF\}$. 
		
		\item Let $\cI^{\cG}_L$ be the independent sets of the $\cG$-matroid of $n,d$ and $L$. 
		
		\item For all $j \in [d]$, let $\cI_{L,j}$ be the independent sets of the  $(L,j)$-partition matroid of $\gr$.
	\end{itemize}
\end{definition}

A main property of the above reduction is that an Empty Set instance $I$ is a ``yes''-instance if and only if at least one of its reduced instances is an $\ell$-MI ``yes''-instance. This is formalized in the next lemma. 

\begin{lemma}
	\label{lem:reductionHelp}
	Let $\ell \geq 3$, let $n \in \N$ such that $n^{\frac{1}{\ell-1}} \in \N$ and $n^{\frac{1}{\ell-1}} \geq 2$. Also, let $k \in [n-1] \setminus \{1\}$. Then, a given \textnormal{ES} instance $I = (n,k,\cF)$ is a ``yes''-instance for if and only if there is an \textnormal{SU} matrix $L \in \N^{N \times (\ell-1)}$ such that $R_L(I)$ is a ``yes''-instance for $\ell$\textnormal{-MI}.
\end{lemma}

\begin{proof}
	Let $\pi:[n] \rightarrow \gr$ be the bijection used in the reduction of $I$ (see \Cref{def:Reduced3MI}).	For the first direction, assume that $I$ is a ``yes''-instance. Then, there is $S \in \cF$; thus, $\pi(S) \in \cG$. As $\cF \subseteq \cS_{n,k}$, it follows that $2 \leq |S| = |\pi(S)| = k \leq n-1 = N^d-1$. Therefore, by \Cref{lem:properties2}, there is an SU matrix $L \in \N^{N \times d}$ such that $\pi(S) \in \textnormal{\textsf{bases}}(\cI^{\cG}_{L},\cI_{L,1},\ldots, \cI_{L,d})$. It follows that $R_L(I)$ is a ``yes''-instance for $\ell$-MI.
	
	Conversely, assume that $I$ is a ``no''-instance. Then, for every $S \subseteq [n]$ it holds that $S \notin \cF$. Therefore, for every $T \subseteq \gr$ it holds that $T \notin \cG$. Hence, by \Cref{lem:properties}, for every SU matrix $L \in \N^{N \times d}$ it holds that $\textnormal{\textsf{bases}}(\cI^{\cG}_{L},\cI_{L,1},\ldots, \cI_{L,d}) = \emptyset$. It follows that, for every SU matrix $L \in \N^{N \times d}$, $R_L(I)$ is a ``no''-instance. 
\end{proof}

We use \Cref{lem:reductionHelp} to prove \Cref{thm:3MIMain}. Intuitively, in the proof of  \Cref{thm:3MIMain} we show that given an algorithm $\cA$ for $\ell$-MI (for some $\ell \geq 3$) which uses a relatively small number of queries to the given membership oracles, we can decide an oracle-ES instance $I = (n,k,\cF)$ using a small number of queries to the membership oracle of $\cF$. This is accomplished simply by constructing all the reduced instances of $I$ (and two more corner cases) and verifying if one of them is a ``yes''-instance for $\ell$-MI. Since the number of SU matrices is relatively small compared to the number of possibilities for $\cF$, using \Cref{lem:EmptySet} we obtain a strong lower bound on the number of queries required for such an algorithm $\cA$.

\thmMI*

\begin{proof}

Let $\ell \geq 3$ and $d = \ell -1$. Consider a randomized algorithm $\cA$ which decides oracle $\ell$-MI in $f(m) \geq 1$ queries to the given membership oracles, where $m$ is the number of elements in the ground set of the given matroids, and $f$ is some  function. Using $\cA$ we construct an algorithm $\cB$ that decides the oracle-Empty Set (oracle-ES) problem. Let $n \in \N$ such that $n^{\frac{1}{\ell-1}} \in \N$ and $n^{\frac{1}{\ell-1}} \geq 2$. Also, let $I = (n,k,\cF)$ be an oracle-ES instance with a universe of size $n$. 
Define Algorithm $\cB$ on input $I$ as follows.

\begin{enumerate}
	\item {\bf If} $k = n$ {\bf return} that $I$ is a ``yes''-instance if and only if $[n]\in \cF$.  
	\item {\bf If} $k = 1$ decide if $I$ is a ``yes'' or ``no'' instance by exhaustive enumeration over $\cS_{n,k}$. 
	\item {\bf Else:}
	\begin{enumerate}
		\item {\bf For all} SU matrices $L \in \N^{N \times d}$ {\bf do:}
		\begin{itemize}
			\item Call $\cA$ on the reduced oracle-$\ell$-\textnormal{MI} instance $R_L(I)$. 
			\item {\bf If} $\cA$ returns that $R_L(I)$ is a ``yes''-instance: {\bf return} that $I$ is a ``yes''-instance. 
		\end{itemize} 
		\item  {\bf return} that $I$ is a ``no''-instance. 
	\end{enumerate}
\end{enumerate}

By \Cref{obs:oracle} and \Cref{def:Reduced3MI}, for every SU matrix $L \in \N^{N \times d}$ we can define a membership oracle for all matroids of $R_L(I)$ that uses at most one query to the membership oracle of $\cF$. We note that $\cB$ may be random only if $\cA$ is random. We analyze below the correctness of the algorithm. 

\begin{claim}
	\label{claim:Correctnes3MI'}
	Algorithm $\cB$ correctly decides  every \textnormal{oracle-ES} instance $I = (n,k,\cF)$ with universe of size $n$ with probability at least $\frac{1}{2}$. 
\end{claim}

\begin{claimproof}
	If $k = n$ or $k = 1$ then $\cB$ trivially decides $I$ correctly with probability $1$. Assume then that $2 \leq k \leq n-1$. 
	For the first direction, suppose that $I$ is a ``yes''-instance. 
	Then, by \Cref{lem:reductionHelp}, there is an SU matrix $L \in \N^{N \times d}$ such that  $R_L(I)$ is a ``yes''-instance. 
	This implies that $\cA$ returns that $R_L(I)$ is a ``yes''-instance with probability at least $\frac{1}{2}$. Hence, $\cB$ returns that $I$ is a ``yes''-instance with probability at least $\frac{1}{2}$. 
	Conversely, assume that $I$ is a ``no''-instance. 
	Then, by \Cref{lem:reductionHelp}, for every SU matrix $L \in \N^{N \times d}$ it holds that $R_L(I)$ is a ``no''-instance. Thus, for every SU matrix $L \in \N^{N \times d}$, we have that $\cA$ returns that $R_L(I)$ is a ``no''-instance with probability $1$. Thus, $\cB$ returns that $I$ is a ``no''-instance with probability $1$.
\end{claimproof}

We give below an analysis of the query complexity of the algorithm. Let $N = n^{\frac{1}{d}}$. 

\begin{claim}
	\label{claim:QueryComplexity'}
	The number of queries used by $\cB$ on input $I$ is at most $(N+1)^{N \cdot d} \cdot f\left(n\right)$. 
\end{claim}

\begin{claimproof}
	If $k = 1$ or $k = n$, then the number of queries performed by $\cB$ on input $I$ is bounded by ${n \choose 1} = n = N^d \leq (N+1)^{N \cdot d} \cdot f(n)$.  Otherwise, by the query complexity guarantee of $\cA$, the number of queries is bounded by the number of SU matrices $L \in \N^{N \times d}$ multiplied by $f(n)$. By \Cref{def:SU'}, the number of SU matrices $L \in \N^{N \times d}$ is at most $(N+1)^{N \cdot d}$, which gives the statement of the claim. 
\end{claimproof}

By Claims~\ref{claim:Correctnes3MI'} and~\ref{claim:QueryComplexity'}, $\cB$ is a randomized algorithm that decides every oracle-ES instance $I = (n,k,\cF)$ in at most $(N+1)^{N \cdot d} \cdot f\left(n\right)$ queries to the membership oracle of $\cF$.  Recall the binomial identity $\sum^{n}_{k = 0} {n \choose k} = 2^{n}$. By the pigeonhole principle, there is $k \in \{0,1,\ldots, n\}$ such that ${n \choose k} \geq \frac{2^n}{n+1}$. Therefore, by \Cref{lem:EmptySet}, 
\begin{equation*}
	\label{eq:f(n)}
	(N+1)^{N \cdot d} \cdot f\left(n\right) \geq \frac{2^{n-1}}{(n+1)}. 
\end{equation*} 
We can now give a lower bound on the number of queries used by $\cA$:

\begin{equation}
	\label{eq:xlog1'}
	\begin{aligned}
		f(n) \geq \frac{2^{n-1}}{(n+1) \cdot (N+1)^{N \cdot d}} 
		= \frac{2^{n-1}}{2^{\log(n+1)} \cdot 2^{\log \left((N+1)^{N \cdot d} \right)}} = 2^{ \Big(n-1-\log(n+1)-N \cdot d \cdot \log (N+1) \Big)}.
	\end{aligned}
\end{equation}

Observe that 
\begin{equation}
	\label{eq:xlog2'}
	\begin{aligned}
		\log(n+1)+N \cdot d \cdot \log (N+1) \leq{} & 	\log\left(n^2\right)+N \cdot d \cdot \log \left(N^2\right)
		\\ ={} & 	2 \cdot \log\left(n\right)+n^{\frac{1}{d}} \cdot d \cdot \log \left(n^{\frac{2}{d}}\right)
		\\ ={} & 	2 \cdot \log (n)+2 \cdot n^{\frac{1}{d}} \cdot \log (n) 
		\\ \leq{} & 	4 \cdot n^{\frac{1}{d}} \cdot \log (n).
		\\ ={} & 	4 \cdot n^{\frac{1}{\ell-1}} \cdot \log (n).
	\end{aligned}
\end{equation} The first inequality holds since $N = n^{\frac{1}{\ell-1}} \geq 2$. Thus, by \eqref{eq:xlog1'} and \eqref{eq:xlog2'} we have
$$f(n) \geq 2^{\Big(n-1-4 \cdot n^{\frac{1}{\ell-1}} \cdot \log (n)\Big)} \geq 2^{\Big(n-5 \cdot n^{\frac{1}{\ell-1}} \cdot \log (n)\Big)}.$$ 
We conclude that there is no randomized algorithm which decides oracle $\ell$-MI instances with $n$ elements in fewer than $2^{\Big(n-5 \cdot n^{\frac{1}{\ell-1}} \cdot \log (n)\Big)}$ queries. 
\end{proof}

		\section{Monotone Local Search}
		\label{sec:monotone}

In this section, we present our generalization of Monotone Local Search. The  generalization is used to attain a faster than brute force algorithm for oracle $\ell$-MI (\Cref{thm:faster_MI}), and a lower bound for the running time of parameterized algorithms for oracle $\ell$-MI (\Cref{thm:param_lb}). 
We give a generic result which can be applied to the wide class of implicit set problems, capturing oracle $\ell$-MI as a special case. 
 In \Cref{sec:monotone_def}, we provide the basic definitions needed to state the results along with the formal statements of the main theorems.   In \Cref{sec:real_monotone}  we present the Monotone Local Search Algorithm and prove its correctness. In \Cref{sec:prop_phi} we provide several lower bounds for the running time of  Monotone Local Search for several specific cases. Finally, in \Cref{sec:monotone_app} we show how to derive  an extension algorithm for $\ell$-MI from a  parameterized algorithm for the problem.

\subsection{Formal Definitions and Results}
\label{sec:monotone_def}

Our results apply to problems which can be cast as {\em implicit set problem}. An implicit set problem $\cP$ is a set of instances of the form $(E,\cF,B,\oracle)$ where $E$ is a finite arbitrary set, $\cF\subseteq 2^E$ is a collection of subset of $E$, $B\in \{0,1\}^*$ is the {\em encoding} of the instance and $\oracle:2^{E}\times \mathbb{N}\rightarrow \{0,1\}$ is the {\em oracle} of the instance.
The instance $(E,\cF,B,\oracle)$  is a ``yes'' instance if $\cF\neq \emptyset$ and is a ``no'' instance if $\cF=\emptyset$. 
 An algorithm for $\cP$  is given the encoding $B$ as input and runs in the computational model of an  oracle Turing machine in which the oracle is $\oracle$. 
 The objective of the algorithm is to determine if $\cF\neq\emptyset$, while the set $\cF$ is only implicitly given to the algorithm through the encoding $B$ and the oracle $\oracle$.

We say that an implicit set problem $\cP$ is {\em polynomial time computable} if the following two conditions hold:
\begin{itemize}
	\item There is an algorithm in an oracle Turing machine  model
	such that for every $(E,\cF, B,\oracle)\in \cP$, given $B$ as the input and $\oracle$ as the oracle, computes the set $E$ in time $\poly(|B|)$. 
	\item There is an algorithm in an oracle Turing machine model such that for every $(E,\cF, B,\oracle)\in \cP$, given $B$ and $S\subseteq E$  as the input and $\oracle$ as the oracle,  decides whether $S\in \cF$ in time  $\poly(|B|)$. 
\end{itemize}
For simplicity, we assume that all the implicit set problems considered in this section are polynomial time computable.

 The $\ell$-matroid intersection problem can be easily cast as an implicit set problem $\cP_{\ell\textnormal{-MI}}$. 
For every $\ell$-matroid intersection instance $(E,\cI_1,\ldots, \cI_2)$ we add to $\cP_{\ell-\textnormal{MI}}$ the instance $(E,\cF, B,\oracle)$ where $\cF = \bases(\cI_1,\ldots,\cI_\ell)$, the encoding $B$ is an arbitrary encoding of the set $E$ and $\oracle(S,j) =1$ if $S\in \cI_j$ and $\oracle(S,j)=0$ otherwise. That is, $\oracle$ is a unified representation of the membership oracles for $\cI_1,\ldots,\cI_{\ell}$. We assume the encoding $B$ satisfies $|E|\leq |B| \leq |E|^2$. It can be easily verified that  $\cP_{\ell\textnormal{-MI}}$ is also polynomial time computable. The set $E$ can be easily computed in polynomial time   as it is explicitly  encoded in $B$, and given $S\subseteq E$ it can be determined whether $S\in\bases(\cI_1,\ldots, \cI_{\ell})$ by verifying $S\in \cI_j$ and is maximal independent set for every $j\in[\ell]$  via  queries to the oracle.

The definition of an implicit set problem given above is a natural generalization of the {\em implicit set systems} defined in \cite{FGLS19}.
 The difference is that in implicit set systems the instance does not contain an oracle. 
As the main focus of this paper is the oracle $\ell$-MI problem, this extension is required.
 We further note that all  the result in \cite{FGLS19} also hold for implicit set problems which involve oracles.

The purpose of Monotone Local Search is to convert an extension algorithm for  $\cP$ into an exponential time algorithm for $\cP$. For many problems, such as $\ell$-matroid intersections, Vertex Cover and Feedback Vertex Set,  such algorithms can be easily derived from parameterized algorithms. A list of classic problems which can be cast as implicit set problems  and for which an extension algorithm can be derived from a parameterized algorithm  can be found in \cite{FGLS19}. 

The definition of extension algorithms relies on the notion of {\em extensions}. 
\begin{definition}[$\ell$-extension]
	Let  $\cP$ be an implicit set problem and $(E,\cF,B,\oracle)\in \cP$ be an instance of the problem.   Also, let $X\subseteq E$, $S\subseteq E\setminus X$  and $\ell\in \mathbb{N}$. We say that {\em $S$ is an $\ell$-extension of $X$} if $X\cup S\in \cF$ and $\abs{S}= \ell$. 
\end{definition}

For an implicit set problem $\cP$ and a function $g:\mathbb{N}\rightarrow \mathbb{N}$, a {\em random extension algorithm} for $\cP$ of time $g$ is an algorithm in which runs in an oracle Turing machine model.
 The input for the algorithm is $B\in\{0,1\}^*$, the encoding of an instance $(E,\cF,B,\oracle)\in P$,  a set  $X\subseteq \cF$ and $\ell\in \mathbb{N}$.  Furthermore, the oracle of the Turing machine is $\oracle$. The algorithm either returns $S\subseteq E$ or a special symbol $\perp \notin 2^{E}$ and satisfies the following properties. 
 \begin{itemize}
 	\item
 If $X$ has an $\ell$-extension then the algorithm returns an $\ell$-extension of $X$ with probability at least $1/2$. 
 \item If the algorithm returns a set $S\subseteq E $ then $S$ is an $\ell$-extension of $S$. 
\end{itemize}
Furthermore, the algorithm runs in time $g(\ell)\cdot \poly(|B|)$. 

Indeed, we can show that  a parameterized algorithm for oracle $\ell$-MI 
which runs in time  $g(k)\cdot \poly(|E|)$, implies an extension algorithm for $\cP_{\ell\textnormal{-MI}}$ of time $g$. The main idea in the construction of the extension algorithm is to run the parameterized algorithm on the instance $(E\setminus X,\cI_1/X,\ldots, \cI_{\ell}/X)$ following some trivial checks. 
To maintain the flow of the discussion focused on  Monotone Local Search, we provide the proof of the following lemma in \Cref{sec:monotone_app}.
\begin{restatable}{lemma}{MIext}
	\label{lem:MI_ext}
	Let $g:\mathbb{N}\rightarrow \mathbb{N}$ be an arbitrary function. 
	If there is a randomized parameterized algorithm for oracle $\ell$-MI which runs in time $g(k)\cdot \poly(|E|)$ then there is a random extension algorithm for $\cP_{\ell\textnormal{-MI}}$ of time $g$, where $k$ is the rank of the first matroid of the input instance and $E$ is the ground set of the input instance. 
\end{restatable}
In particular, by \Cref{lem:MI_ext} there is a random extension algorithm for $\cP_{\ell\textnormal{-MI}}$ of time $g(\ell)=c^{\ell^2}$ for some $c>1$ as a consequence of the algorithm for \cite{huang2023fpt}. 

The main result of \cite{FGLS19} is the following.
\begin{theorem}[\cite{FGLS19}]
	\label{thm:FGLS}
	Let $\cP$ be a implicit set problem which has a random extension algorithm of time $g(\ell)=c^{\ell}$ for some $c\geq 1$. Then there is an randomized algorithm for $\cP$ which runs in time $\left(2-\frac{1}{c}\right)^{|E|} \cdot \poly(|B|)$ for every instance $(E,\cF,B,\oracle)\in \cP$. 
\end{theorem}
As already mentioned, the results in \cite{FGLS19} are stated for implicit set systems which do not involve oracles. However, the result in \cite{FGLS19}  trivially generalizes to \Cref{thm:FGLS}.  

By \Cref{thm:FGLS} and \Cref{thm:3MIMain} it follows that there is no random extension algorithm for $\cP_{\ell\textnormal{-MI}}$ of time $g(\ell)= c^{\ell}$, for any $c\geq 1$, as such algorithm would imply by \Cref{thm:FGLS} a $\left(2-\frac{1}{c}\right)\cdot \poly(|B|)$ algorithm for $\cP_{\ell\textnormal{-MI}}$  which contradicts \Cref{thm:3MIMain}. As a consequence, by \Cref{lem:MI_ext}, there is no random parameterized algorithm for oracle $\ell$-MI which runs in time $c^k\cdot \poly(|E|)$. 
Our goal is to provide a variant of \Cref{thm:FGLS} which gives  a stronger lower bound for the running time of a parameterized algorithm for oracle $\ell$-MI. Furthermore, \Cref{thm:FGLS} cannot be used together with the random extension algorithm for  $\cP_{\ell\textnormal{-MI}}$ of time $g(\ell)=c^{\ell^2}$, as it the theorem only support extensions algorithms with running time of the form $g(\ell)=c^{\ell}$. Thus, the second goal of our extension of \Cref{thm:FGLS} is to provide a variant which can utilize this extension algorithm to get an algorithm for oracle $3$-MI which is faster than brute force.

Our generalization of monotone local search needs to be able to optimize its use of the extension oracle. This is done by solving a discrete optimization problem over the function $g$, the running time of the extension algorithm. We therefore require $g$ to be a function for which the optimization problem can be solved efficiently, as stated by the following definition. 
\begin{definition}
	\label{def:optimizable}
	We say a function $g$ is {\em optimizable} if for every $n\in \mathbb{N}$ and $0\leq k\leq n$ the value $$\argmin_{0 \leq t \leq k} \frac{\binom{n}{t}}{\binom{k}{t}} \cdot g(k-t)$$
	can be computed in polynomial time in $n$.  
\end{definition}
In case $g(n)$ can be computed in polynomial time in $n$ then $g$ is trivially optimizable as  $\argmin_{0 \leq t \leq k} \frac{\binom{n}{t}}{\binom{k}{t}} \cdot g(k-t)$ can be computed in polynomial time by iterating over all possible values of $t$. In other cases, such as $g(n)=2^{2^n}$ the value of $g$ cannot be represented using polynomial number of bits (in standard representation), and hence a slight sophistication is required in order to show that $g$ is still optimizable. 

Our extension for monotone local search is the following. 
\begin{theorem}[Monotone Local Search]
	\label{thm:monotone_local_search}

		Let $\cP$ be a implicit set problem which has a random extension algorithm of time $g$ such that $g$ is optimizable  and $g(0)=1$. Then there is an randomized algorithm for $\cP$ which runs in time
	$$
	2^{|E|-\Phi_g(|E|)} \cdot \poly(|B|), 
	$$
	where $E$ is the ground set of the instance, $B$ is the encoding of  the instance, and
	\begin{equation}
		\label{eq:phi_def}
		\Phi_g(n) = \max \left\{ \ell \cdot \log \left(\frac{n}{4\cdot \ell}\right) -\log g(\ell) ~\big|~0\leq \ell \leq n/4,~\ell\in \mathbb{N}
		\right\}.
	\end{equation}
\end{theorem}
The proof of \Cref{thm:monotone_local_search} is given in \Cref{sec:real_monotone}.  
While \Cref{thm:monotone_local_search} can be used with $g(n)=c^n$, the running time it  provides in such cases is inferior to the running time guaranteed by \Cref{thm:FGLS}. 

We  also  provide estimations for $\Phi_g(n)$ for several special cases which are relevant to our applications. 
\begin{restatable}{lemma}{phiklogk}
	\label{lem:phi_bound_klogk}
	Define $g_{\alpha}(\ell)=\floor{2^{\alpha \cdot\ell \cdot\log (\ell) }}$ every $\alpha>0$. Then
	$\Phi_{g_{\alpha}}(n)  = \Omega\left(n^{\frac{1}{1+\alpha}} \right)$.
\end{restatable}
The following lemma is an immediate consequence of \Cref{thm:monotone_local_search} and \Cref{lem:phi_bound_klogk}. 
\begin{lemma}
\label{lem:implicit_klogk}
Let $\cP$ be an implicit set problem which has an extension algorithm of time $g(\ell)=  \floor{2^{\alpha\cdot \ell \cdot \log \ell}}$ for some $\alpha>0$. 
Then $\cP$ has a randomized algorithm which runs in time $2^{|E|-\Omega\left(|E|^{\frac{1}{1+\alpha}}\right)}\cdot \poly(|B|)$, where $E$ is the ground set of the instance and  $B$ is the encoding of  the instance.
\end{lemma}
We use \Cref{lem:implicit_klogk} to show the following. 
\thmlowerBoundParameterized*
\begin{proof}
Assume towards a contradiction that there is a parameterized algorithm for oracle $\ell$-MI which runs in time $\floor{2^{\alpha\cdot k\log(k)}}\cdot \poly(|E|)$ for where $\alpha=\ell-2-\eps$. Then by \Cref{lem:MI_ext} there is a random extension algorithm for $\cP_{\ell\textnormal{-MI}}$ of time $g_{\alpha}(\ell) =\floor{2^{\alpha \cdot \ell\log(\ell)}}$. Hence, by \Cref{lem:implicit_klogk} there is an algorithm for  $\cP_{\ell\textnormal{-MI}}$  which runs in time 
$$
	2^{|E|-f(|E|) }\cdot \poly(|B|)  \leq
	 2^{|E|-f(|E|)  + c\cdot \log(|E|) }
$$
for some $c\geq 0$ where $f(n) =\Omega(n^{\frac{1}{1+\alpha}})$. By \Cref{thm:3MIMain} for every $n\in\mathbb{N}$ such that $n\geq 2^{\ell-1}$ and $n^{\frac{1}{\ell-1}}\in \mathbb{N}$,  it must hold that 
$$
n-f(n) + c\cdot \log(n)\geq n-5\cdot n^{\frac{1}{\ell-1}}\log n .
$$
Therefore, 
$$
f(n) \leq 5\cdot n^{\frac{1}{\ell-1}}\log(n)  +c\cdot \log(n) = O(n^{\frac{1}{\ell-1}}\cdot \log n ),
$$
for infinitely many values of $n\in \mathbb{N}$. 
Since $f(n)= \Omega\left(n^{\frac{1}{1+\alpha}} \right)$,  we have $\frac{1}{1+\alpha}\leq  \frac{1}{\ell-1}$. That is, $\ell-2-\eps = \alpha\geq \ell -2$. A contradiction to $\eps>0$.
 \end{proof}
 
 We use the following estimation for $\Phi$ to attain an algorithm for oracle $\ell$-MI which is significantly faster than brute force. 
\begin{restatable}{lemma}{phiksquare}
	\label{lem:phi_bound_kquare}
	Let $g(\ell)=\floor{2^{\alpha\cdot \ell^2}}$ for some $\alpha>0$. Then
	$\Phi_g(n)  = \Omega\left(  \log^2(n)  \right)$.
\end{restatable}
The following lemma is an immediate consequence of \Cref{thm:monotone_local_search} and \Cref{lem:phi_bound_kquare}.
\begin{lemma}
	\label{lem:implicit_ksquare}
Let $\cP$ be an implicit set problem which has an extension algorithm of time $g(\ell)=  \floor{2^{\alpha\cdot \ell^2}}$ for some $\alpha>0$. Then $\cP$ has a randomized algorithm which runs in time $2^{|E|-\Omega(\log^2(|E|))}\cdot \poly(|B|)$, where $E$ is the ground set of the instance and  $B$ is the encoding of  the instance.
\end{lemma}

We use \Cref{lem:implicit_ksquare} to prove \Cref{thm:faster_MI}.

\thmalg*
\begin{proof}
Since there is a parameterized algorithm for $\ell$-matroid intersection which runs in time $c^{k^2}\cdot \poly(|E|)$ for some $c>1$, then by \Cref{lem:MI_ext} there is a random extension algorithm for $\cP_{\ell\textnormal{-MI}}$ of time $g(\ell)= \floor{c^{\ell^2}}$. 
Therefore by \Cref{lem:implicit_ksquare} there is a randomized algorithm for $\ell$-MI which runs in time 
$$ 2^{|E|-\Omega(\log^2(|E|))}\cdot \poly(|B|) = 2^{|E|-\Omega(\log^2(|E|))}.
$$
\end{proof}

Finally, we show that for every function $g$ it holds that the running time of the monotone local search algorithm is better than $\frac{2^n}{\poly(n)}$.  
\begin{restatable}{lemma}{asymphi}
	\label{lem:phi_asym}
	For every function $g:\mathbb{N}\rightarrow \mathbb{N}$ it holds that 
	$\Phi_g(n) = \omega(\ln n )$. 
\end{restatable}

\Cref{lem:phi_asym} together with  \Cref{thm:monotone_local_search} immediately imply the following.
\begin{lemma}
	\label{lem:better_than_brute}
	Let $\cP$ be an implicit set problem such that there is an random extension algorithm for $\cP$ of time $g$ for an arbitrary optimizable function  $g$. Then, there is an algorithm for $\cP$ of time $2^{|E|-\omega(\log|E|)}\cdot \poly(B)$. 
\end{lemma}
Every implicit set problem $\cP$ can be solved in time $\approx 2^{|E|}$ using brute force enumeration over all subsets $S$ of $E$.  \Cref{lem:better_than_brute} essentially asserts that if $\cP$ has an extension algorithm then it can be solved faster than brute force.

The proofs of \Cref{lem:phi_bound_klogk,lem:phi_bound_kquare,lem:phi_asym} are given in \Cref{sec:prop_phi}.

\subsection{Monotone Local Search: The Algorithm}
\label{sec:real_monotone}

The algorithm we use to prove \Cref{thm:monotone_local_search} is nearly identical to the  {\em Monotone Local Search} algorithm of Fomin, Gaspers, Lokshtanov  and Saurabh \cite{FGLS19}. The key distinction is in the   analysis which considers  different running times for the extension algorithm. 
Let $\cP$ be an implicit set problem. We assume $\cP$ is fixed throughout the section. Our goal is to design an algorithm for $\cP$ for which there is a random extension algorithm  $\ext$ of time $g$, where the function $g$ satisfies the conditions in \Cref{thm:monotone_local_search}. We further assume $\ext$ and $g$ are fixed throughout the section. 

Recall that an instance of $\cP$ is of the form $(E,\cF, B,\oracle)$ where the encoding $B$ is given to the algorithm as input and the algorithm can access $\oracle$ as an oracle. 

Consider the case in which  $\cF\neq \emptyset$,  then there is $S^*\in\cF$. Let  that $\abs{S^*}=k$.  The algorithm guesses $k$ by iterating over all values from $0$ to $|E|$. Now,  if one samples a uniformly random set $X$ of $[n]$ of size $t$, then with probability $\frac{\binom{k}{t}}{\binom{n}{t}}$ it holds that $X\subseteq S^*$ (there are $\binom{n}{t}$ subsets of $[n]$ of size $t$, and $\binom{k}{t}$ of these subsets only contain items in $S^*$).
Furthermore, if $X\subseteq S^*$ then $\ext(B,X,k-t)$ returns a set  $S$ such that $X\cup S\in \cF$ with probability at least $\frac{1}{2}$.
Therefore, if the algorithm randomly samples a set $X$ of size $t$ and invokes $\ext(B,X,k-t)$ for $2\cdot \frac{\binom{n}{t}}{\binom{k}{t}}$ many times, with probability at least $\frac{1}{2}$ one of the calls for $\ext$ returns a $(k-t)$-extension  (and not the symbol $\perp$). In such a case, the algorithm can safely  conclude that $\cF\neq \emptyset$. Otherwise, the algorithm may return that $\cF=\emptyset$, and be right with high probability. 
The value of $t$ is simply selected so the overall running time is minimized. 

The pseudo-code of the algorithm is  given in~\Cref{alg:monotone_local_search}. The algorithm finds the optimal value for $t$ and initiates multiple calls to the $\sample$ procedure, given in \Cref{alg:sample}.  
The procedure  depends on the random extension algorithm $\ext$ for $\cP$. 
The procedure samples a subset $X\subseteq E$ of size $t$ and attempts to extend it using  $\ext$. 

\begin{algorithm}[t]
	\caption{$\sample(B,k,t)$}
	\SetKwInOut{Input}{input}
	\SetKwInOut{Output}{output}
	\SetAlgoNlRelativeSize{0}
	\label{alg:sample}
	
	\Input{ Encoding $B$ of an instance $(E,\cF,B,\oracle) \in \cP$, number $k,t\in \mathbb{N}$ and  oracle access to $\oracle$. } 
	
	Compute the set $E$ from $B$
	
 	Sample a uniformly random subset $X$ of $E$ of size $t$ \label{sam:sample}
	
	$S\leftarrow \ext(B,X,k-t )$  \label{sam:ext}\\

	\Return S	

\end{algorithm}

\begin{algorithm}[b]
	\caption{Monotone Local Search for $\cP$}
	\SetKwInOut{Input}{input}
	\SetKwInOut{Output}{output}
	\SetAlgoNlRelativeSize{0}
	\label{alg:monotone_local_search}
	
	\Input{ Encoding $B$ of an instance $(E,\cF,B,\oracle) \in \cP$, oracle access to $\oracle$. 
	 } 
	
	Define $\cL \leftarrow \emptyset$\\
	
	Compute the set $E$ of the instance, and let $n\leftarrow |E|$. \\
	
	\For{$k$  from $0$ to $n$\label{mls:loop}} {	
	$t \leftarrow \argmin_{0 \leq t \leq k} \frac{\binom{n}{t}}{\binom{k}{t}} \cdot g(k-t)$\label{mls:t}

	Run $\cL \leftarrow \cL\cup \sample(B,k,t)$ for $2\cdot \frac{\binom{n}{t}}{\binom{k}{t}}$  times \label{mls:sample} 
	}
	{\bf If} $\cL=\{\perp\}$  {\bf then} return ``no'', {\bf else} return ``yes''\label{mls:return}
\end{algorithm}

 \Cref{thm:monotone_local_search}  is an immediate consequence of \Cref{lem:mls_correct,lem:mls_cost} presented below.
\begin{restatable}{lemma}{mlscorrect}
	\label{lem:mls_correct}
	\Cref{alg:monotone_local_search} is a random algorithm for $\cP$. 
\end{restatable}
\begin{restatable}{lemma}{mlscost}
	\label{lem:mls_cost}
	\Cref{alg:monotone_local_search} runs in time  $f2^{|E|-\Phi_g(|E|)}\cdot  \poly(|B|)$. 
\end{restatable}
We give the proofs \Cref{lem:mls_correct,lem:mls_cost} in \Cref{sec:mls_correct,sec:mls_cost}, Respectively.

\subsubsection{Correctness} 
\label{sec:mls_correct}
The correctness of the algorithm follows from the argument presented in the beginning of this section. 
\mlscorrect*
\begin{proof}
	
Let $(E,\cF,B,\oracle)$ be an  instance of $\cP$. Also, define $n=|E|$. 
In order to show the correctness of \Cref{alg:monotone_local_search}, we first show basic properties of $\sample$ (\Cref{alg:sample}).  

First, we consider the case in which $\cF=\emptyset$.
\begin{claim}
	\label{claim:sample_empty}
	If $\cF=\emptyset $, then  $\sample(B,k,t) = \perp$ for every $0\leq k\leq n$ and $0\leq t\leq k$  (with probability $1$). 
\end{claim}
\begin{claimproof}
	Consider an execution of $\sample$ with $B$, $k$ and $t$  as its input. 
	Let $X$ be the set defined in \Cref{sam:sample} and let $S$ be the variable defined in \Cref{sam:ext}. 
Since $\cF=\emptyset$ then $S=\perp$  as $X$ does not have an $(k-t)$-extension (no set has an $(k-t)$-extension).  Therefore, the algorithm returns $\perp$. 
\end{claimproof}

On the other hand, in case $\cF\neq \emptyset$ we show the following.
\begin{claim}
	\label{claim:sample_prob}
	If $S^*\in \cF$ and $\abs{S^*}=k $, then  for every $0\leq t \leq k$ it holds that $\sample(B,k,t)\neq \perp$ with probability at least  $\frac{1}{2}\cdot \frac{ \binom{k}{t}}{\binom{n}{t}}$. That is, 
	 $\Pr\left(\sample(B,k,t) \neq \perp\right) \geq\frac{1}{2}\cdot \frac{ \binom{k}{t}}{\binom{n}{t}}$.
\end{claim}
\begin{claimproof}
	
Consider an execution of $\sample(B,k,t)$ and let $X$ and  $S$ be the values of the variables defined in \Cref{sam:sample,sam:ext}, respectively. 

Let $\cT= \{ T\subseteq S^* \,|\, \abs{T}=t \}$ be all the subsets of $S^*$ of size $t$. It trivially holds that $\abs{\cT}= \binom{k}{t}$. For every set $T\in \cT$  it holds that $S^*\setminus T$ is a $(k-t)$-extension of $T$. Therefore, as $\ext$ is an extension algorithm we have,
$$
\Pr(\ext(B,T,k-t)\neq \perp) \,\geq\, \frac{1}{2}
$$
Therefore,
\begin{equation}
	\label{eq:prob_given_X}
\Pr( S\neq \perp\,|\,X=T) =\Pr(\ext(B,T,k-t)\neq \perp) \,\geq\, \frac{1}{2},
\end{equation}
for every $T\in \cT$.  
By the above, 
$$
\begin{aligned}
\Pr(S\neq \perp) &= \sum_{T\subseteq E:~|T|=t} \Pr(X=T)\cdot\Pr(S\neq \perp\,|\,X=T) \\
&\geq\sum_{T\in \cT} \Pr(X=T)\cdot\Pr(S\neq \perp\,|\,X=T)   \\
&\geq \sum_{T\in \cT} \frac{1}{\binom{n}{t}}\cdot\frac{1}{2} \\
	&= \frac{1}{2}\cdot \frac{ \binom{k}{t}}{\binom{n}{t}} ,
	\end{aligned}
$$
The second inequality holds as $\Pr(X=T) = \frac{1}{\binom{n}{t}}$ for every set $T\subseteq E$ of size $t$ and by \eqref{eq:prob_given_X}.  The last equality holds as $\abs{\cT} = \binom{k}{t}$. 
\end{claimproof}

Consider an execution of \Cref{alg:monotone_local_search}  with input $B$ and the oracle $\oracle$.  Let $\cL$  be the value of the respective variable at the end of the algorithm's execution.  
Consider the following cases.
\begin{itemize}
	\item
In case $\cF=\emptyset$, then by  \Cref{claim:sample_empty} all the calls for $\sample(B,k,t)$ return $\perp$, therefore, $\cL=\{\perp\}$ and the algorithm correctly returns ``no'' in \Cref{mls:return}. 
\item
In case $\cF\neq \emptyset$, 
then there is $S^*\in \cF$. Consider the iteration of \Cref{mls:loop} in which $k=|S^*|$ and let $t$ be the value found in \Cref{mls:t} in the specific iteration.  
Then $\cL=\{\perp\}$ only if all the calls for $\sample(B,k,t)$ in the specific iteration  return $\perp$.  Since those calls are independent we have,
$$
\begin{aligned}
\Pr(\cL= \{\perp\})&= \Pr\left(\textnormal{all $2\cdot \frac{\binom{n}{t}}{\binom{k}{t}}$  calls for $\sample(B,k,t)$ returned $\perp$  } \right)\\
&= \left(\Pr( \sample(B,k,t,) = \perp) \right)^{2\cdot \frac{\binom{n}{t}}{\binom{k}{t}}} \\
&\leq \left(1- \frac{1}{2}\cdot \frac{\binom{k}{t}}{\binom{n}{t}} \right)^{2\cdot \frac{\binom{n}{t}}{\binom{k}{t}}} \\
&\leq e^{-1},
\end{aligned}
$$
where the first inequality follows from \Cref{claim:sample_prob}, and the second inequality from $\left(1-\frac{1}{x}\right)^x \leq e^{-1}$ for all $x\geq 1$. Therefore $\Pr(\cL\neq \{\perp\})\geq 1-e^{-1} \geq \frac{1}{2}$.  We can therefore conclude  that the algorithm returns ``yes'' with probability at least $1/2$.
\end{itemize} 
Overall, we showed that \Cref{alg:monotone_local_search} is indeed a randomized algorithm for  $\cP$. 
 \end{proof}
 
 \subsubsection{Running Time} 
 \label{sec:mls_cost}

 Our next goal is to  prove \Cref{lem:mls_cost}. That is, our  objective is to show that the running time  of  \Cref{alg:monotone_local_search} is $2^{|E|- \Phi_g(|E|)}\cdot \poly(|B|)$ where $\Phi_g$ is the function defined in \eqref{eq:phi_def}. 
 We achieve this objective in two steps.
In \Cref{lem:rt_by_psi} we  show that the running  of the algorithm is bounded by  $\Psi_g(|E|)\cdot \poly(B)$ where the function  $\Psi_g$  is defined by
 \begin{equation}
 	\label{eq:psi_def}
 \Psi_g(n) = \max_{0\leq k \leq n} \,\min_{0\leq t\leq k} \frac{\binom{n}{t} }{\binom{k}{t}} \cdot g(k-t).
 \end{equation}
 Then, in \Cref{lem:psi_bound} we show that $\Psi_g(n) \leq 2^{n-\Phi_g(n) +\log(n)}$. Recall we assumed $g$ is a fixed function, $g(0)=1$ and $g$ is optimizable. 
 \begin{lemma}
 	\label{lem:rt_by_psi}
 \Cref{alg:monotone_local_search} runs in time $\Psi_g(|E|)\cdot  \poly(|B|)$.
 \end{lemma}
 \begin{proof}
 	Let $(E,\cF, B,\oracle)$ be an instance of $\cP$ and 
 	consider an execution of \Cref{alg:monotone_local_search} with the input $B$  and the oracle $\oracle$. Also, let $n=|E|$. 
 	We first bound the running time of each of the  iterations  of the loop in  \Cref{mls:loop} of \Cref{alg:monotone_local_search} separately. Consider the  $k$-th iteration of the loop. 
 	Let $t$ be the value found in \Cref{mls:t} of \Cref{alg:monotone_local_search} in the specific iteration. 
 	We note that the computation of $t$ can be done in time $\poly(|E|)\leq \poly(|B|)$ since  $g$ is optimizable (\Cref{def:optimizable}).
 	Then, each  of the $2\cdot \frac{\binom{n}{t}}{\binom{k}{t}}$ calls for $\sample$ in \Cref{mls:sample} runs in time  $g(k-t)\cdot \poly(|B|)$. This is  due to the use of $\ext$ plus a polynomial number of operations  in $|B|$ for computing $|E|$ and sampling the set $X$ in \Cref{sam:sample} of \Cref{alg:sample}. 
 	Therefore, the total running time  of \Cref{mls:sample} in the specific iteration is 
 	$$
 	\begin{aligned}
 	2\cdot \frac{\binom{n}{t}}{\binom{k}{t}}\cdot g(k-t) \cdot \poly(|B|) &=\poly(|B|)\cdot  \min_{0\leq t'\leq k} 	 \frac{\binom{n}{t'}}{\binom{k}{t'}}\cdot g(k-t'),
 \end{aligned}$$
 where the first equality follows from the selection of $t$ in \Cref{mls:t}.  This implies that the total running time of the specific iteration is also
 $$
 \poly(|B|)\cdot  \min_{0\leq t'\leq k} 	 \frac{\binom{n}{t'}}{\binom{k}{t'}}\cdot g(k-t'). 
 $$
 
 Therefore,  the overall running time of the loop in \Cref{mls:loop} of \Cref{alg:monotone_local_search} is 
 $$
 \begin{aligned}
 \sum_{k=0}^{n} \poly(|B|)\cdot  \min_{0\leq t'\leq k} 	 \frac{\binom{n}{t'}}{\binom{k}{t'}}\cdot g(k-t')\, &\leq \,
 	\poly(|B|)\cdot \max_{0\leq k' \leq n}\min_{0\leq t'\leq k'} 	 \frac{\binom{n}{t'}}{\binom{k'}{t'}}\cdot g(k'-t') 
 	\\
 	&=\,  \Psi_g(n)\cdot \poly(|B|),
 	\end{aligned}
$$
 
 	Therefore the whole execution of \Cref{alg:monotone_local_search}  runs in time  $\Psi_g(|E|)\cdot \poly(|B|)$. 
 \end{proof}

 We use standard estimators of binomial coefficients to upper bound $\Psi_g$.  
Define 
$$
\entropy{x} = -x \cdot \log (x )- (1-x)\cdot \log (1-x)
$$
as the binary entropy function. 
  We use the following entropy based estimation for binomial coefficients (Example 11.1.3 in \cite{CJ06}):

\begin{equation}
	\label{eq:binom}
	\frac{1}{n+1}\cdot 2^{n\cdot \entropy{ \frac{k}{n}}}\,\leq\,
	\binom{n}{k}\,\leq \,2^{n\cdot \entropy{ \frac{k}{n}}}.
\end{equation}
We also use the following inequality for the entropy function:
\begin{equation}
	\label{eq:h_by_x}
	\frac{\entropy{x}}{x} = \frac{-x \cdot\log (x) -(1-x )\cdot \log(1-x)}{x} = -\log x -\frac{ (1-x )\cdot \log (1-x)}{x}  \geq -\log (x).
\end{equation}
The last equality follows from $ (1-x)\cdot \log (1-x)<0$.

\begin{lemma}
	For every $n\in \mathbb{N}$ it holds that 
	\label{lem:psi_bound}
	$$ 
		\Psi_g(n) \leq 2^{ n-\Phi_g(n) +\log(n)}. 
	$$
\end{lemma}
\begin{proof}
	We can write $\Psi_g$ as the maximum between two functions. Define, 
	$$
	A(n )=	\max_{0\leq k < \frac{1}{4}\cdot n  \textnormal{ or } \frac{3}{4}\cdot n < k\leq n}\, \min_{0\leq t\leq k} \frac{\binom{n}{t} }{\binom{k}{t}} \cdot g(k-t) 
	$$
	and 
	\begin{equation}
		\label{eq:B_def}
	B(n)= 	\max_{\frac{1}{4}\cdot n\leq k \leq \frac{3}{4}\cdot n }\, \min_{0\leq t\leq k} \frac{\binom{n}{t} }{\binom{k}{t}} \cdot g(k-t).
	\end{equation}
	The function $A$ considers values of $k$ which are far from $\frac{1}{2}\cdot n$.  In these cases  selecting $t=k$ suffices to attain the upper bound in the lemma. The function $B$ considers values of $k$ which may be close to $\frac{1}{2}\cdot n$, for which we use a more subtle argument to upper bound $\Psi_g$. 
	By  \eqref{eq:psi_def} we have 
	\begin{equation}
		\label{eq:psi_split_sum}
		\Psi_g(n) = \max_{0\leq k \leq n}\, \min_{0\leq t\leq k} \frac{\binom{n}{t} }{\binom{k}{t}} \cdot g(k-t) = 
		\max 
		\left\{  A(n)
		,B(n)
		\right\}.
	\end{equation}
	We bound  $A(n)$ and $B(n)$ separately. 
	
	\begin{claim}
		\label{claim:A_bound}
	$A(n)\leq 2^{0.85\cdot n}$
	\end{claim}
	\begin{claimproof}
		Let $0\leq k \leq n$ such that $k\leq \frac{1}{4}\cdot n $ or $k\geq \frac{3}{4}\cdot n$. Then by  selecting $t=k$ we have 
		$$
			\begin{aligned}
				 \min_{0\leq t\leq k} \frac{\binom{n}{t} }{\binom{k}{t}} \cdot g(k-t)
				\leq \frac{\binom{n}{k} }{\binom{k}{k}} \cdot g(0)
				 =   \binom{n}{k}
				\leq 2^{n\cdot \entropy{\frac{k}{n}}} 
				\leq 2^{0.85\cdot n },
			\end{aligned}
$$	where the second equality follows from \eqref{eq:binom},  the equality uses $g(0)=1$, and the last inequality follows from $\entropy{x}\leq 0.85$ for $x\in [0,1/4]\cup [3/4,1]$ as  $\entropy{x}$ is increasing in $[0,1/2]$, decreasing in $[1/2,1]$, and  $\entropy{1/4}=\entropy{3/4} < 0.85$.   Therefore, 
$$
	A(n )=	\max_{0\leq k < \frac{1}{4}\cdot n  \textnormal{ or } \frac{3}{4}\cdot n < k\leq n}\, \min_{0\leq t\leq k} \frac{\binom{n}{k} }{\binom{n-t}{k-t}} \cdot g(k-t)\leq  2^{0.85\cdot n }
	$$
	\end{claimproof}
	
	Next, we bound $B(n)$. 
	\begin{claim}
		\label{claim:B_bound}
		$B(n)\leq 2^{n-\Psi_g(n) +\log(n)}$
	\end{claim}
	\begin{claimproof}
	Let $\frac{1}{4}\cdot n \leq k \leq \frac{3}{4}\cdot n$. Then, 
\begin{equation}
	\label{eq:B_first}
		\begin{aligned}
	\log&\left(	 \min_{1\leq t\leq k} \frac{\binom{n}{t} }{\binom{k}{t}} \cdot g(k-t)\right) =\min_{0\leq t\leq k} \left( \log \binom{n}{t}  - \log\binom{k}{t} + \log\left( g(k-t)\right)  \right)\\
		&\leq  \min_{0\leq t\leq k} \left( n\cdot \entropy{\frac{t}{n}} -k\cdot \entropy{\frac{t}{k}} +\log(k+1) +\log g(k-t)\right)\\
		&\leq \min_{0\leq t\leq k} \left( n -k\cdot \entropy{\frac{t}{k}} +\log(n) +\log g(k-t)\right)
		\end{aligned}
\end{equation}
where the first inequality follows from \eqref{eq:binom}, and the second holds as $k< n$ and $\entropy{x} \leq 1$ for all $0\leq x\leq 1$. 
By \eqref{eq:h_by_x}, for every $0\leq t\leq k$,  we have
$$
-k\cdot \entropy{\frac{t}{k}}  = -(k-t)\cdot \frac{k}{k-t}\cdot  \entropy{\frac{k-t}{k}} \leq (k-t)\cdot \log\left( \frac{k-t}{k}\right) \leq (k-t)\cdot \log\left(4\cdot\frac{k-t}{ n}\right),
$$
where the first inequality follows from $\entropy{x}=\entropy{1-x}$ and the last inequality holds as $k\geq \frac{1}{4}\cdot n$.
Incorporating the above into \eqref{eq:B_first} we get
\begin{equation}
	\label{eq:B_second}
	\begin{aligned}
		\log\left(	 \min_{1\leq t\leq k} \frac{\binom{n}{t} }{\binom{k}{t}} \cdot g(k-t)\right) 
		&\leq \min_{0\leq t\leq k} \left( n -k\cdot \entropy{\frac{t}{k}} +\log(n) +\log g(k-t)\right)
		\\
			&\leq \min_{0\leq t\leq k} \left( n  +(k-t)\cdot \log\left(4\cdot \frac{k-t}{ n}\right)+\log(n) +\log g(k-t)\right)
	\end{aligned}
	\end{equation}

	We use several auxiliary functions in order to bound $B(n)$. Define $$\varphi(\ell,n) =  \ell \cdot \log \left( \frac{n}{4\cdot \ell}\right) -\log g(\ell) $$ and 
	\begin{equation}
		\label{eq:p_def}p(n)= \argmax_{0\leq \ell \leq n/4} \varphi(\ell,n).\end{equation}
	Therefore, by
	\eqref{eq:phi_def}  it holds that
	 \begin{equation}
		\label{eq:phi_to_varphi}
		\Phi_g(n) = \varphi(p(n), n).\end{equation}  
		We use $t=k-p(n)$ to upper bound the expression in \eqref{eq:B_second}. Observe that since $p(n)\leq \frac{n}{4}\leq k$ it holds that $k-p(n)\geq 0$. That is, 
	$$\begin{aligned}
			\log\left(	 \min_{1\leq t\leq k} \frac{\binom{n}{t} }{\binom{k}{t}} \cdot g(k-t)\right) 
		&\leq \min_{0\leq t\leq k} \left( n  +(k-t)\cdot \log\left(4\cdot \frac{k-t}{ n}\right)+\log(n) +\log g(k-t)\right)\\
		&\leq   n  +p(n)\cdot \log\left( \frac{4\cdot p(n)}{ n }\right)+\log(n) +\log g(p(n))\\
		&\leq n-\varphi\left(p(n),n\right) +\log(n).
	\end{aligned}
$$
Overall, we showed that
\begin{equation}
	\label{eq:for_all_k}	\log\left(	 \min_{1\leq t\leq k} \frac{\binom{n}{t} }{\binom{k}{t}} \cdot g(k-t)\right)  \leq  n-\varphi\left(p(n),n\right) +\log(n)
\end{equation}
for all $\frac{n}{4}\leq k \leq \frac{3\cdot n}{4}$. 

By the above we have, 
$$
	B(n)= 	\max_{\frac{1}{4}\cdot n\leq k \leq \frac{3}{4}\cdot n }\, \min_{0\leq t\leq k} \frac{\binom{n}{t} }{\binom{k}{t}} \cdot g(k-t) \leq 2^{n-\varphi\left(p(n),n\right) +\log(n)}=2^{n -\Phi_g(n) + \log (n)},
$$
where the first equality is by \eqref{eq:B_def}, the first inequality is by \eqref{eq:for_all_k}, and the last equality is by \eqref{eq:phi_to_varphi}.

	\end{claimproof}

	By  \Cref{claim:A_bound}, \Cref{claim:B_bound} and \eqref{eq:psi_split_sum} we have,
	\begin{equation}
		\label{eq:psi_max}
	\Psi_g(n) = 
	\max 
	\left\{  A(n)
	,B(n)
	\right\} = \max\left\{ 2^{0.85\cdot n},~2^{ n-\Phi_g(n) +\log(n)}\right\}.
	\end{equation}
	To complete the proof we use the following claim.
	\begin{claim}
		\label{claim:phi_little_o}
		 $\Phi_g(n) \leq 0.15 \cdot n$ for all $n \in \N$.
	\end{claim}
	\begin{claimproof}
		Assume towards contradiction that there is $n\in \mathbb{N}$ such that $\Phi_g(n) > 0.15\cdot n$. Then, 
		$$
			0.15\cdot n< \Psi_g(n)= \varphi(p(n),n)= p(n)\cdot \log\left( \frac{n}{4\cdot p(n)} \right) - \log g(p(n)) \leq p(n)\cdot \log\left( \frac{n}{4\cdot p(n)} \right).
		$$
		Therefore,
$$		 \frac{p(n)}{n}\cdot \log\left( \frac{n}{4\cdot p(n)}\right)  > 0.15.$$
However, the function $r(x)=\frac{1}{x} \log \left(\frac{1}{4}x\right)$ has a global maximum of $0.132\approx \frac{1}{4\cdot e \cdot \ln 2} <0.15$ at $x=4e\approx 10.87$. Therefore, there is no $n$ for which $\Phi_g(n) > 0.15n$.

	\end{claimproof}
	
	By \eqref{eq:psi_max} and  \Cref{claim:phi_little_o} it follows that 
	$$
	\Psi_g(n)  = \max\left\{ 2^{0.85\cdot n},~2^{ n-\Phi_g(n) +\log(n)}\right\} =  2^{ n-\Phi_g(n) +\log(n)} 
	$$
\end{proof}

\Cref{lem:mls_cost} is an immediate consequence of \Cref{lem:rt_by_psi,lem:psi_bound}. 

\subsection{Properties of $\Phi_g$}
\label{sec:prop_phi}
In this section we prove \Cref{lem:phi_bound_klogk}, \Cref{lem:phi_bound_kquare} and \Cref{lem:phi_asym}.  Those are the lemmas  which bound the value of $\Phi_g(n)$, as defined in \eqref{eq:phi_def}, for various functions $g$.
 All the proofs follow from elementary calculus. 

\phiklogk*
\begin{proof}
	Recall that $\Phi_g(n)$ is defined by a maximization problem over a variable $\ell$ (see \eqref{eq:phi_def}). We define a function $\ell(n)$ and attain a lower bound for $\Phi_g(n)$ by setting $\ell=\ell(n)$. 
For every $n\in \mathbb{N}$ define $\ell(n) = \floor{\frac{1}{c}\cdot n^{\frac{1}{1+\alpha}}}$, where $c= 2^{\frac{3}{1+\alpha}}$.  Observe  that 
\begin{equation}
	\label{eq:klogk_ell}
\frac{1}{c}\cdot n^{\frac{1}{1+\alpha}} -1 \leq \ell(n)\leq  \frac{1}{c}\cdot n^{\frac{1}{1+\alpha}}.
\end{equation}
Therefore, as $\alpha>0$,  there is $N>0$ such that $3\leq \ell(n)\leq \frac{1}{4}\cdot n $ for all $n\geq N$. 
Hence, for every $n>N$ we have
\begin{equation}
	\label{eq:klogk_phi_decomp}
	\Phi_{g_{\alpha}}(n)  = \max_{0\leq \ell \leq \frac{1}{4}\cdot n }\left( \ell \cdot \log\left( \frac{n}{4\cdot \ell}\right) - \log \left( g_{\alpha}(\ell)\right)\right) 
	\geq \ell(n) \cdot \log\left( \frac{n}{4\cdot \ell(n)}\right) - \log \left( g_{\alpha}(\ell(n))\right).
\end{equation}
By \eqref{eq:klogk_ell} we have,
\begin{equation}
	\label{eq:klogk_phi_first}
	\begin{aligned}
\ell(n) \cdot \log\left( \frac{n}{4\cdot  \ell(n)}\right) &\geq  \left(\frac{ n^{\frac{1}{1+\alpha}}}{c} -1\right) \cdot \log\left(\frac{n}{4\cdot \frac{ n^{\frac{1}{1+\alpha}}}{c}}\right) \\
&= \frac{ n^{\frac{1}{1+\alpha}}}{c} \cdot \left( \log\left(c\right) - \log(4) +\frac{\alpha}{1+\alpha }\cdot  \log(n) \right) -\log\left(\frac{c}{4}\cdot n^{\frac{\alpha}{1+\alpha}}\right).
\end{aligned}
\end{equation}
Furthermore, it holds that
\begin{equation}
	\label{eq:klogk_phi_second}
	\begin{aligned}
 \log \left( g_{\alpha}(\ell(n))\right) &\leq\alpha \cdot \ell(n)\cdot \log(\ell(n)) \\
 &\leq \alpha \cdot \frac{ n^{\frac{1}{1+\alpha}}}{c} \cdot \log \left(\frac{n^{\frac{1}{1+\alpha}}}{c} \right)\\&=   \frac{ n^{\frac{1}{1+\alpha}}}{c} \cdot \left( \frac{\alpha}{1+\alpha}\cdot \log(n) -\alpha\cdot \log(c)\right),
 \end{aligned}
\end{equation}
where the first inequality follows from $\log(g_{\alpha}(\ell)) \leq \log\left(\floor{2^{\alpha \ell \cdot \log(\ell )}  }\right)  \leq \alpha\cdot \ell \log(\ell)$ and  the second inequality  holds as $x\log(x)$ is increasing in $[1,\infty)$ and due to \eqref{eq:klogk_ell}. 

Plugging both \eqref{eq:klogk_phi_first} and \eqref{eq:klogk_phi_second} into \eqref{eq:klogk_phi_decomp} we get,
$$
\begin{aligned}
	\Phi_{g_{\alpha}}(n)  
&\geq \ell(n) \cdot \log\left( \frac{n}{4\cdot \ell(n)}\right) - \log \left( g_{\alpha}(\ell(n))\right)\\
&\geq \frac{ n^{\frac{1}{1+\alpha}}}{c} \cdot \left( \log\left(c\right) -\log(4) +\frac{\alpha}{1+\alpha } \cdot \log(n) \right) -\log\left(\frac{c}{4}\cdot n^{\frac{\alpha}{1+\alpha}}\right)\\
&~~~~~~~~~~~~ -  \frac{ n^{\frac{1}{1+\alpha}}}{c} \cdot \left( \frac{\alpha}{1+\alpha}\cdot \log(n) -\alpha\cdot \log(c)\right)\\
&= \frac{ n^{\frac{1}{1+\alpha}}}{c} \cdot \bigg( (1+\alpha)\cdot \log(c) - \log(4)\bigg)-
\log\left(\frac{c}{4}\cdot n^{\frac{\alpha}{1+\alpha}}\right).\\
&=  \frac{ n^{\frac{1}{1+\alpha}}}{c} -\log\left(\frac{c}{4}\cdot n^{\frac{\alpha}{1+\alpha}}\right).\\
&= \Omega \left( n^{\frac{1}{1+\alpha}} \right),
 \end{aligned}
$$
 where second equality holds as $c=2^{\frac{3}{1+\alpha}}$.
\end{proof}

We proceed to the proof of \Cref{lem:phi_bound_kquare}. The proof follows a similar outline to the proof of \Cref{lem:phi_bound_klogk}.
\phiksquare*
\begin{proof}
	Define $\beta = \max\{\alpha,1\}$. Then $g(\ell) \leq \floor{2^{\beta\cdot \ell^2}}$. 
For every 
$n \in \mathbb{N}$ 
define  $\ell(n) = \floor{\frac{1}{4\cdot\beta}\cdot \log(n)}$. 
Since $\log(x)<x$  for all $x>0$ it holds that $0\leq \ell(n)\leq \frac{1}{4}\cdot n$. Therefore, 
\begin{equation}
	\label{eq:ksqaure_phi_decomp}
	\Phi_g(n)  = \max_{0\leq \ell \leq \frac{1}{4}\cdot n }\left( \ell \cdot \log\left( \frac{n}{4\cdot \ell}\right) - \log \left( g(\ell)\right)\right) 
	\geq \ell(n) \cdot \log\left( \frac{n}{4\cdot \ell(n)}\right) - \log \left( g(\ell(n))\right)
\end{equation}
for all $n\in \mathbb{N}$. Furthermore,
\begin{equation}
	\label{eq:ksqaure_first} 
\begin{aligned}
\ell(n) \cdot \log\left( \frac{n}{4\cdot \ell(n)}\right)  &= \floor{\frac{1}{4\cdot \beta }\cdot \log(n)}  \cdot \log  \left( \frac{n}{4\cdot \floor{\frac{1}{4\cdot \beta}\cdot \log(n) }} \right) 
\\
&\geq \left(\frac{1}{4\cdot \beta}\cdot \log(n)-1\right)\cdot \left( \log(n) - \log\left(\frac{1}{\beta} \cdot \log(n)\right) \right) \\
&\geq \frac{1}{4\cdot \beta}\cdot\log^2(n) - \log(n)  - \frac{1}{4}\cdot \log(n)\cdot \log \log(n)
\end{aligned}
\end{equation}
and
\begin{equation}
	\label{eq:ksqaure_second}
\begin{aligned}
\log(g(\ell(n)))  \leq \alpha \cdot \left(\ell(n )\right)^2 \leq \beta \cdot \frac{1}{16\cdot \beta^2} \cdot \log^2(n) \leq \frac{1}{16\cdot \beta }\cdot \log^2(n).
\end{aligned}
\end{equation}
By \eqref{eq:ksqaure_phi_decomp}, \eqref{eq:ksqaure_first} and \eqref{eq:ksqaure_second} we have
$$
\begin{aligned}
	\Phi_g(n)  
&\geq \ell(n) \cdot \log\left( \frac{n}{4\cdot \ell(n)}\right) - \log \left( g(\ell(n))\right)\\
&\geq  \frac{1}{4\cdot \beta}\cdot \log^2(n) - \log(n)  - \frac{1}{4}\cdot \log(n)\cdot \log \log(n)- \frac{1}{16\cdot \beta}\cdot \log^2(n)\\
&= \frac{3}{16\cdot \beta} \cdot \log^{2}(n)- \log(n)  - \frac{1}{4}\cdot \log(n)\cdot \log \log(n)\\
& = \Omega( \log^2(n)),
\end{aligned}
$$
which completes the proof. 
\end{proof}

Finally, we prove \Cref{lem:phi_asym} which deals with an arbitrary function $g$. 
\asymphi*
\begin{proof}
Define $h(\ell) = \max\big\{ 2^{\ell^2} , g(\ell)\big\}$ and 
\begin{equation}
	\label{eq:elln_def}
\ell(n)  =\max\left\{ 1\leq \ell\leq \frac{n}{4}\,\middle|\,\ell\in \mathbb{N} \textnormal{ and } \frac{\log(h(\ell))}{\ell}\leq \frac{1}{2}\cdot \log(n) \right\}.
\end{equation}
It can be easily observed that there is $N\in \mathbb{N}$ such that $\ell(n)$ is well defined for all $n>N$. 
For every $n>N$ it holds that,
\begin{equation}
	\label{eq:phi_asym}
	\begin{aligned}
	\Phi_g(n)  &= \max_{0\leq \ell \leq \frac{1}{4}\cdot n }\left( \ell \cdot \log\left( \frac{n}{4\cdot \ell}\right) - \log \left( g(\ell)\right)\right) 
\\
		&	\geq \ell(n) \cdot \log\left( \frac{n}{4\cdot \ell(n)}\right) - \log \left( g(\ell(n))\right)\\
		& \geq  \ell(n) \cdot \log\left( \frac{n}{4\cdot \ell(n)}\right) - \frac{1}{2} \cdot \ell(n)\cdot \log(n)\\
		&=\frac{1}{2}\cdot \ell(n) \cdot \log(n) - \ell(n)\cdot \log(4\cdot \ell(n)),
	\end{aligned}
\end{equation}
where the second inequality follows from $$\log(g(\ell(n)) )\leq \log(h(\ell(n))) \leq \frac{1}{2}\cdot \ell(n)\cdot \log(n)$$ by \eqref{eq:elln_def}.  By \eqref{eq:elln_def} we also have 
\begin{equation}
	\label{eq:asym_ell_bound}
\ell(n) = \frac{\ell^2(n)}{\ell(n)} \leq \frac{\log(h(\ell(n)))}{\ell(n)} \leq \frac{1}{2}\cdot \log(n)
\end{equation}
for all $n>N$. By \eqref{eq:phi_asym} and \eqref{eq:asym_ell_bound} we have
\begin{equation}
	\label{eq:phi_asyn2}
	\Phi_g(n)  \geq \frac{1}{2}\cdot \ell(n) \cdot \log(n) - \ell(n)\cdot \log(4\cdot \ell(n)) \geq  \frac{1}{2}\cdot \ell(n) \cdot \log(n) -\log(n) \cdot \log( 4\cdot \ell(n))
\end{equation}
	for all $n>N$. By \eqref{eq:phi_asyn2} we have
	$$
	\begin{aligned}
	\liminf_{n\rightarrow \infty } \frac{\Phi_g(n)}{\log(n)} &\geq 	\liminf_{n\rightarrow \infty }\frac{\frac{1}{2}\cdot \ell(n) \cdot \log(n) -\log(n) \cdot \log( 4\cdot \ell(n)) }{\log(n)}\\&
	=\liminf_{n\rightarrow \infty }\left( \frac{1}{2}\cdot \ell(n) - \log(4\cdot \ell(n))\right). 
	\end{aligned}
	$$ 
	By \eqref{eq:elln_def} it holds that $\ell(n)$ is an increasing function of $n$. Furthermore, for every $M>0$ and $n\geq 2^{2\cdot \frac{\log(h(\ceil{M}))}{\ceil{M}}}$ it holds that $\ell(n) \geq \ceil{M}\geq M$. Therefore \begin{equation}
		\label{eq:lim_ell}\lim_{n\rightarrow \infty} \ell(n) =\infty.
		\end{equation}
		 Define  $\kappa(x)  = \frac{1}{2}\cdot x - \log(4\cdot x)$ and observe that 
		 \begin{equation}
		 \label{eq:lim_kappa}
		 \lim_{x\rightarrow \infty} \kappa(x)= \infty.
		 \end{equation}
		 
		 By the above inequalities we have,
		 \begin{equation}
		 	\label{eq:phi_asym3}
		 	\liminf_{n\rightarrow \infty } \frac{\Phi_g(n)}{\log(n)} 
		 \geq \liminf_{n\rightarrow \infty }\left( \frac{1}{2}\cdot \ell(n) - \log(4\cdot \ell(n))\right) =\liminf_{n\rightarrow \infty} \kappa(\ell(n)) = \infty,
		 \end{equation}
		 where the first inequality is by \eqref{eq:phi_asyn2} and the equality follows from \eqref{eq:lim_ell} and \eqref{eq:lim_kappa}.  By \eqref{eq:phi_asym3} we have $\Phi_g(n) = \omega(\log(n))$, as required. 
\end{proof}

\subsection{An Extension Algorithm for $\cP_{\ell\textnormal{-MI}}$}
\label{sec:monotone_app}

We are left to show how an extension algorithm for 
$\cP_{\ell\textnormal{-MI}}$ can be derived from a random parameterized algorithm for oracle $\ell$-MI.

\MIext*
The proof of \Cref{lem:MI_ext}  relies on the contraction operation of matroids which has been defined in \Cref{sec:matroids}.  We note that if $(E,\cI)$ is a matroid and $S\cup X$ is a basis of $(E,\cI)$ such that $S\cap X=\emptyset$ then $S$ is a basis of $(E,\cI)/X=(E\setminus X, \cI/X)$, the contraction of $(E,\cI)$ by $X$. This is also true in the opposite direction - if $S$ is a basis of $(E,\cI)/X=(E\setminus X, \cI/X)$ then $X\cup S$ is a basis of $(E,\cI)$.  
The core idea in the proof of \Cref{lem:MI_ext} is that if a set $X$ has an $\ell$-extension $S$, that is, $X\cup S$ is a common basis of $(E,\cI_1),\ldots, (E,\cI_\ell)$  and $|S|=\ell$ then $S$ is  a common basis of $(E,\cI_1)/X,\ldots, (E,\cI_\ell)/X$.

\begin{proof}[Proof of \Cref{lem:MI_ext}]

Let $\cA$ be a randomized parameterized algorithm for oracle $\ell$-MI which runs in time $g(k)\cdot \poly(|E|)$.
  We assume that $\cA$ returns a common basis of the matroids if there is one. 
We define an extension algorithm $\cD$ for $\cP_{\ell\textnormal{-MI}}$.
 Recall that an instance $(E,\cF, B,\oracle)\in \cP_{\ell\textnormal{-MI}}$ is associated with an $\ell$-matroid intersection instance $(E,\cI_1,\ldots, \cI_{\ell})$ where $\cF=\bases(\cI_1,\ldots, \cI_{\ell})$ and $\oracle$ acts as a unified membership oracle for the sets $\cI_1,\ldots, \cI_{\ell}$. 
 	The input for $\cD$ is $B$, from which $E$ can be computed in polynomial time, a set $X\subseteq E$, and $\ell\in \mathbb{N}$.
 	 Furthermore, $\cD$  has access to the oracle $\oracle$.  The objective of the algorithm is to return an $\ell$-extension of $X$ if one exists, or return $\perp$ if it does not find one. We define $\cD$ as follows.
	\begin{enumerate}
		\item Compute the rank $r$ of $(E,\cI_1)$. \label{ext:rank}
				\item If $|X|+\ell\neq r$ then return $\perp$. \label{ext:sum}
		\item If $X\notin \bigcap_{j=1}^{\ell} \cI_j$ then return $\perp$. \label{ext:intersection}
		\item Run $\cA$ on the instance $(E\setminus X, \cI_1/X,\ldots, \cI_{\ell}/X)$. \label{ext:runA}
		\item If $\cA$ returned a common basis $S$ then return $S$, otherwise return $\perp$. 	\label{ext:final}
 	\end{enumerate}
 	
 	The rank of a matroid can be computed in polynomial time given a membership oracle for the matroid, thus Step~\ref{ext:rank}
 	 can be computed in polynomial time (for example, by finding an arbitrary inclusion-wise maximal independent set - a basis). Furthermore, we note that a membership oracle for $\cI_j/X$ can be trivially be emulated using a membership oracle for $\cI_j$, hence it is possible to run $\cA$ on the instance $(E\setminus X, \cI_1/X,\ldots, \cI_{\ell}/X)$. 
 	
	We show $\cD$ is indeed a random extension algorithm for $\cP_{\ell\textnormal{-MI}}$.
	\begin{claim}
		If $\cD$ returns a set $S$ then $S$ is an $\ell$-extension of $X$
	\end{claim}
	\begin{claimproof}
	If the algorithm returns a set $S$ then the set $S$ has been returned by $\cA$ in Step \ref{ext:runA}. 
	Therefore $S$ is a common basis of $(E\setminus X, \cI_1/X),\ldots ,(E\setminus X, \cI_{\ell}/X)$, which implies that $X\cup S$ is a common basis of $(E,\cI_1),\ldots, (E,\cI_\ell)$. That is, $X\cup S\in \cF$. 
	Furthermore, since the algorithm did not return on Step~\ref{ext:sum} we have  $|X|+\ell=r$ and $\abs{X\cup S} =r$  as the size of a basis is the rank of the matroid. Therefore, we also have $|S|=\ell$. That is, $S$ is an $\ell$ extension of $X$. 
	\end{claimproof}
	\begin{claim}
	If there is an $\ell$-extension of $X$ then $\cD$ returns a set $S$ with probability at least $1/2$.
	\end{claim}
	\begin{claimproof}
		If there in $\ell$-extension $S$ of $X$ then $S\cup X$ is a common basis of $(E,\cI_1),\ldots, (E,\cI_{\ell})$ and $\abs{S}=\ell$. In particular, $\ell + \abs{X}=\abs{S\cup X} = r$, therefore the algorithm does not return on Step~\ref{ext:sum}. Furthermore, $X\in \bigcap_{j=1}^{\ell} \cI_j$ by the hereditary property of matroids, therefore the algorithm also does not return on Step~\ref{ext:intersection}. Hence, the algorithm reaches Step~\ref{ext:runA}. Since $X\cup S$ is common basis of $(E,\cI_1),\ldots, (E,\cI_{\ell})$, it follows that $S$ is a common basis of $(E\setminus X, \cI_1/X),\ldots ,(E\setminus X, \cI_{\ell}/X)$. Therefore, $\cA$ returns a set with probability at least $1/2$, which means that $\cD$ returns a set with probability at least $1/2$ as well. 
	\end{claimproof}
	
	Finally, we note that Steps~\ref{ext:rank},~\ref{ext:sum},~\ref{ext:intersection} can be implemented in polynomial time in $B$. 
	If the algorithm reaches Step~\ref{ext:runA} then 
	the rank of $(E\setminus X, \cI_1/X)$ is $r-|X|=\ell$. Therefore, 
	 Step~\ref{ext:runA} runs in time 
	$
	g(\ell)\cdot \poly(|E|)= g(\ell)\cdot \poly(|B|)
	$.
	Hence, the total running time of $\cD$ is $g(\ell)\cdot \poly(|B|)$. 
	Overall, we showed that $\cD$ is a randomized extension algorithm for $\cP_{\ell\textnormal{-MI}}$ of time $g$. 
\end{proof}

		\section{Discussion and Open Problems}
	\label{sec:discussion}
	In this paper, we obtained an almost tight running time lower bound for $\ell$-matroid intersection for any $\ell \geq 3$, even if randomization is allowed. In addition, we generalized the monotone local search technique of Fomin et al. \cite{FGLS19} for a wider class of parameterized algorithms, and used it (i) to obtain an algorithm for $\ell$-MI that is faster than brute force by a super-polynomial factor, and (ii) to derive a parameterized lower bound for $\ell$-MI. A few intriguing questions remain open: 


\paragraph{Strengthening our Results} 
We derived several running time lower bounds for $\ell$-MI. While these bounds are close to the state-of-the-art algorithmic results, they do not fully match. More specifically, we showed a lower bound of $2^{n-5 \cdot n^{\frac{1}{\ell-1}} \cdot \log n}$ while we presented only a $2^{n-\Omega\left(\log^2 n \right)}$ algorithm. Also,
we showed a running time lower bound of $2^{(\ell-2-\eps) \cdot k \cdot \log k} \cdot \poly(n)$ for parameterized $\ell$-MI, while the best known algorithm of Huang and Ward \cite{huang2023fpt} runs in time $c^{k^2} \cdot \poly(n)$. Observe that these gaps are related: a faster parameterized algorithm would imply a faster exponential time algorithm, and a stronger exponential time lower bound would imply a stronger lower bound in the parameterized setting. This is due to our generalization of monotone local search. Can the algorithms be improved? Alternatively, can the lower bounds be strengthened?

\paragraph{Monotone Local Search} Our generalization of the monotone local search allows only for randomized algorithms.
While the original monotone local search paper \cite{FGLS19} presents also a derandomization of the technique, the derandomization came with a cost of $o(n)$ overhead in the exponent, i.e., running time of $\left(2-\frac{1}{c}\right)^{n+q(n)}$ instead of $\left(2-\frac{1}{c}\right)^{n}\cdot \poly(n)$, where $q(n)= \Theta\left(\frac{n}{\log n }\right) =o(n)$. While this  overhead can be claimed to be  insignificant in the setting of \cite{FGLS19}, in our setting this may be harmful to the extent that the overall runtime becomes higher than brute force. Thus, obtaining a useful derandomized algorithm for our generalization is an interesting direction for futue work.

\paragraph{Beating Brute Force for other Problems} 
We have studied $3$-MI as an example for a problem with no algorithm 
of running time $c^n\cdot \poly(n)$ for any $c<2$, which still admits an algorithm faster than brute force by a factor of $n^{\log n}$, where $n$ is the size of the instance.
The monotone local search technique used to derive this algorithm 
is generic and can be applied to a wide range of problems which admit parameterized algorithms. Are there other such problems, for which the best known  (brute force) exponential-time algorithm can be improved using monotone local search?

\newpage
\paragraph{Acknowledgments:} We thank Lars Rohwedder and Karol Wegrzycki for a helpful discussion on our results.

  \bibliographystyle{alpha}

		\bibliography{bibfile}
		\appendix
		
			\newpage
			\section{Omitted Proofs}
			\label{sec:proofs}
		In this section, we give the proofs missing from the paper body. 

\lemES*
\begin{proof}

	Assume towards a contradiction that there are $n_0 \in \N$, $k_0 \in \{0,1\ldots, n_0\}$, and a randomized algorithm $\cA$ that decides the \textnormal{oracle-ES} problem on a universe of size $n_0$ and cardinality target $k_0$ in strictly fewer than $\frac{{n_0 \choose k_0}}{2}$ queries. 
	Let $F_{\emptyset} = \emptyset$ and consider the ``no''-instance $I_{\emptyset} = (n_0,k_0,\cF_{\emptyset})$ of oracle-ES. In addition, let $q = \frac{{n_0 \choose k_0}}{2}-1$ be the maximum number of queries performed by $\cA$ on instance $I_{\emptyset}$ for any realization of the algorithm.

	Formally, $\cA$ on instance $I_{\emptyset}$ generates a random string of bits $B \in \{0,1\}^{f(I_{\emptyset})}$, where $f$ is some computable function that depends only on $\cA$, and performs a sequence of queries to the oracle.\footnote{Technically, the bit string can be shorter than $f(I_{\emptyset})$. We assume that all realizations of $B$ are of length $f(I_{\emptyset})$ without the loss of generality.} The sequence of queries depends only on $I_{\emptyset}$, $B$, and the previous queries. As a result of the queries, $\cA$ decides $I_{\emptyset}$. We use $B$ to denote the random string of bits and use $b$ to denote concrete realizations of $B$. For some realization $b \in \{0,1\}^{f(I_{\emptyset})}$ of $B$, let $\cQ(b) \subseteq \cS_{n_0,k_0}$ be the set of all sets $S \in \cS_{n_0,k_0}$ queried by $\cA$ on $b$. Since the number of queries of $\cA$ on instance $I_{\emptyset}$  is at most $q$, it follows that $|\cQ(b)| \leq q$ for every $b \in \{0,1\}^{f(I_{\emptyset})}$. Let $\cQ_{\geq \frac{1}{2}}$ be all sets in $\cS_{n_0,k_0}$ that are queried by $\cA$ with probability at least $\frac{1}{2}$, that is:
	\begin{equation}
		\label{eq:Q1/2}
		\cQ_{\geq \frac{1}{2}} = \left\{S \in \cS_{n_0,k_0}~\bigg|~\Pr \left(S \in \cQ(B) \right) \geq \frac{1}{2} \right\}. 
	\end{equation} 
	We show that there is at least one set in $\cS_{n_0,k_0}$ that is not in $\cQ_{\geq \frac{1}{2}}$. 
	\begin{claim}
		\label{claim:Q}
		$\left|\cQ_{\geq \frac{1}{2}} \right|  < {n_0 \choose k_0}$.  
	\end{claim}
	\begin{claimproof}
		By \eqref{eq:Q1/2} it follows that 
		\begin{equation*}
			\begin{aligned}
				\left|\cQ_{\geq \frac{1}{2}} \right|  
				\leq{} & 2 \cdot \sum_{S \in \cS_{n_0,k_0}} \Pr \left(S \in \cQ(B) \right) 
				={}  	 2 \cdot \sum_{S \in \cS_{n_0,k_0}~} \sum_{b \in \{0,1\}^{f(I_{\emptyset})}} \one_{S \in \cQ(b)} \cdot \Pr(B = b),
			\end{aligned}
		\end{equation*} Where $\one_{S \in \cQ(b)}$ is the indicator for the event $S \in \cQ(b)$ for every $S \in \cS_{n_0,k_0}$ and every $b \in \{0,1\}^{f(I_{\emptyset})}$. Then, by changing the order of summation  
		\begin{equation}
			\label{eq:medQ}
			\begin{aligned}
				\left|\cQ_{\geq \frac{1}{2}} \right|   \leq 2 \cdot \sum_{b \in \{0,1\}^{f(I_{\emptyset})}~} \Pr(B = b) \cdot \sum_{S \in \cQ(b)} 1 
				= 2 \cdot  \sum_{b \in \{0,1\}^{f(I_{\emptyset})}~} \Pr(B = b) \cdot |\cQ(b)| .
			\end{aligned}
		\end{equation} 
		
		Since $|\cQ(b)| \leq q$ for every $b \in \{0,1\}^{f(I_{\emptyset})}$, by \eqref{eq:medQ} it follows that 
		
		\begin{equation*}
			\left|\cQ_{\geq \frac{1}{2}} \right|   	\leq 2 \cdot \sum_{b \in \{0,1\}^{f(I_{\emptyset})}~} \Pr (B = b) \cdot q = 2 \cdot q \cdot \sum_{b \in \{0,1\}^{f(I_{\emptyset})}~} \Pr (B = b) = 2  \cdot q
			< 2 \cdot  \frac{{n_0 \choose k_0}}{2} ={n_0 \choose k_0}. 
		\end{equation*}
		The last inequality holds since $q$ is bounded by $\frac{{n_0 \choose k_0}}{2}-1$. By the above, the proof follows. 
	\end{claimproof}
	
	Clearly, $\left| \cQ_{\geq \frac{1}{2}}\right| \in \N$; thus, by \Cref{claim:Q} there is $S^*  \in \cS_{n_0,k_0}$ such that $S^* \notin \cQ_{\geq \frac{1}{2}}$. Let $\cF^* = \left\{  S^* \right\}$ and consider the ``yes''-instance $I^* = (n_0,k_0,\cF^*)$ of oracle-ES. In addition, let $T = \{b \in \{0,1\}^{f(I_{\emptyset})}~|~ S^* \notin \cQ(b)\}$ be all realizations of $B$ satisfying that $S^*$ is not queried by $\cA$ on input $I^*$.  
	Understandably, for all realizations $b \in T$, the oracles of $\cF_{\emptyset}$ and $\cF^*$ return the same output for all queries $S \in \cQ(b)$. Additionally, recall that the next query of the algorithm is determined only by $n_0, k_0$, the realization $b$ of $B$, and the results of previous queries.  
	
	Hence, $\cA$ does not distinguish between the instances $I_{\emptyset}$ and $I^*$ for every realization $b \in T$. Since $I_{\emptyset}$ is a ``no''-instance for oracle-ES and since $\cA$ is a randomized algorithm that decides correctly the oracle-ES problem on a universe of size $n_0$ and cardinality target $k_0$, $\cA$ returns that $I_{\emptyset}$ is a ``no''-instance with probability $1$. On the other hand, $I^*$ is a ``yes''-instance. Thus, $\cA$ decides incorrectly that $I^*$ is a ``yes''-instance for every realization $b \in T$. 
	Using the definition of $T$, we have $\Pr \left( B \in T \right) = \Pr \left(S^* \notin \cQ(B) \right) > \frac{1}{2}$. Therefore, with probability strictly larger than $\frac{1}{2}$ it holds that $\cA$ fails to decide the oracle-ES instance $I^*$. This is 
	a contradiction to the definition of $\cA$. \end{proof} 

We proceed to the proof of \Cref{claim:partitionBases} 
\partitionBases*
 \begin{proof}
	Let $B \in 	\textnormal{\textsf{bases}}(\cI_{L,j})$. Assume towards contradiction that there is $i^* \in [n]$ such that $\left|\left\{ \ve \in B \mid \ve_j = i^*\right\} \right|  \neq L_{i^*,j}$. Since $B \in \textnormal{\textsf{bases}}(\cI_{L,j})$, it follows that $B \in \cI_{L,j}$. This implies that $\left|\left\{ \ve \in B \mid \ve_j = i^*\right\} \right|  < L_{i^*,j}$ and for all $i \in [n]$ it holds that $\left|\left\{ \ve \in B \mid \ve_j = i\right\} \right|  \leq L_{i,j}$. 
	Observe that $L_{i^*,j} \leq n$ since $L$ is an SU matrix. Thus,
	\begin{equation}
		\label{eq:notinB}
		\begin{aligned}
			\left|\left\{  \ve \in B \mid \ve_j = i^* \right\}\right| < L_{i^*,j} \leq n \leq n^{d-1} = 	\left|\left\{  \ve \in \gr \mid \ve_j = i^* \right\}\right|.
		\end{aligned}
	\end{equation}
	By \eqref{eq:notinB}, there is $\ve^* \in \left\{\ve \in \gr \mid \ve_j = i^*\right\} \setminus B$. 
	Hence, $\left|\left\{ \ve \in B+\ve^* \mid \ve_j = i^*\right\} \right|  \leq L_{i^*,j}$. Consequently, for all $i \in [n]$ it holds that $\left|\left\{ \ve \in B+\ve^* \mid \ve_j = i\right\} \right|  \leq L_{i,j}$, implying that $B+\ve^* \in \cI_{L,j}$. Since $|B+\ve^*| > |B|$, we reach a contradiction to $B \in \textnormal{\textsf{bases}}(\cI_{L,j})$.

	For the second direction, let $S \subseteq \gr$ such that for all $i \in [n]$ it holds that $\left|\left\{ \ve \in S \mid \ve_j = i\right\} \right|  = L_{i,j}$. Then, by \eqref{eq:constraints}, $S \in \cI_{L,j}$. Moreover, for all $\ve' \in \gr \setminus S$,  $$\left|\left\{ \ve \in S+\ve' \mid \ve_j = i'\right\} \right|  = L_{i',j}+1,$$
	where $i' = \ve'_j$, implying that $S+\ve' \notin \cI_{L,j}$. Hence, $S$ is a maximal independent set, i.e., $S \in \textnormal{\textsf{bases}}(\cI_{L,j})$. 
\end{proof}

		\newpage

	\end{document}